\documentclass[a4paper,11pt]{article}
\usepackage{amsmath,graphicx,multicol}
\usepackage{amsfonts}
\usepackage{color}
\usepackage{bbm}
\usepackage{cite}
\usepackage{amssymb}
\usepackage{algorithm}
\usepackage{algorithmic,booktabs}
\usepackage[margin=1in]{geometry}
\usepackage{hyperref}
\usepackage{tabularx}
\usepackage[utf8]{inputenc}
\usepackage[english]{babel}

\usepackage{amsthm}
\usepackage{subcaption}
\usepackage[font={footnotesize}]{caption}
\usepackage{url}
\usepackage{mathtools}
\usepackage{comment}
\usepackage{float}
\usepackage[toc,page]{appendix}

\newtheorem{definition}{Definition}[]
\newtheorem{proposition}{Proposition}[]
\newtheorem{theorem}{Theorem}[]

\newtheorem{lemma}[]{Lemma}

\DeclareMathOperator*{\argmax}{arg\,max}
\DeclareMathOperator*{\argmin}{arg\,min}

\DeclareMathOperator{\E}{\mathbb{E}}

\def\n{{\mathbf n}}
\def\y{{\mathbf y}}

\def\C{{\mathbf C}}

\def\P{{\mathbf W}}
\def\I{{\mathbf I}}

\def\r{{\mathbf r}}

\def\x{{\mathbf x}}

\def\A{{\mathbf A}}
\def\Q{{\mathbf Q}}

\def\U{{\mathbf U}}
\def\V{{\mathbf V}}
\def\R{{\mathbb{R}}}

\def\I{{\mathbf I}}
\def\a{{\mathbf a}}

\def\u{{\mathbf u}}

\def\Sb{{\mathbf \Sigma}}

\def\S{{\cal S}}
\def\O{{\cal O}}
\def\T{{\cal T}}
\def\K{{\cal K}}
\def\X{{\cal X}}
\newcommand{\ts}{\textsuperscript}

\newcommand{\Var}{\mathrm{Var}}

\let\oldref\ref
\renewcommand{\ref}[1]{(\oldref{#1})}
\newcommand{\RNum}[1]{\uppercase\expandafter{\romannumeral #1\relax}}
\newcommand{\ra}[1]{\renewcommand{\arraystretch}{#1}}
\makeatletter
\renewcommand{\fnum@figure}{Fig.~\thefigure}
\makeatother

\usepackage{amssymb}

\title{\textbf{Towards Accelerated Greedy Sampling and Reconstruction of Bandlimited Graph Signals}}
\date{}
\vspace{-0.05in}
\author{Abolfazl~Hashemi$^\dagger$, Rasoul~Shafipour$^\ddagger$, Haris~Vikalo, and Gonzalo~Mateos\thanks{Work in this paper was supported in part by the NSF award CCF-1750428 and ECCS-1809327. Abolfazl Hashemi and Haris Vikalo are with the Department of Electrical and Computer Engineering, University of Texas at Austin, Austin, TX 78712, USA. Rasoul Shafipour is with Microsoft, Redmond, WA. Gonzalo Mateos is with the Department of Electrical and Computer Engineering, University of Rochester, Rochester, NY 14627, USA. Part of the results in this paper were presented at the Fourty-Third IEEE International Conference on Acoustics, Speech, and Signal
Processing, Calgary, Canada, April 2018 \cite{hashemi2018sampling} and the Sixth IEEE Global Conference on Signal and Information Processing, Anaheim, California, USA, November 2018 \cite{hashemi2018sampling_global}. $^\dagger$work done while with the Department of Electrical and Computer Engineering, University of Texas at Austin. $^\ddagger$work done while with the Department of Electrical and Computer Engineering, University of Rochester.}}
\begin{document}
\maketitle
\begin{abstract}
\noindent{
{\color{black}
We study the problem of sampling and reconstructing spectrally sparse graph signals where the objective is to select a subset of nodes of prespecified cardinality that ensures interpolation of the original signal with the lowest possible reconstruction error. This task is of critical importance in Graph signal processing (GSP) and while existing methods generally provide satisfactory performance, they typically entail a prohibitive computational cost when it comes to the study of large-scale problems. Thus, there is a need for accelerated and efficient methods tailored for high-dimensional and large-scale sampling and reconstruction tasks. To this end, we first consider a non-Bayesian scenario and propose an efficient iterative node sampling procedure that in the noiseless case enables exact recovery of the original signal from the set of selected nodes. In the case of noisy measurements, a bound on the reconstruction error of the proposed algorithm is established. Then, we consider the Bayesian scenario where we formulate the sampling task as the problem of maximizing a monotone weak submodular function, and propose a randomized-greedy algorithm to find a sub-optimal subset of informative nodes. We derive worst-case performance guarantees on the mean-square error achieved by the randomized-greedy algorithm for general non-stationary graph signals.
}
}
\\\\\noindent
\textbf{Keywords:} graph signal processing, sampling, reconstruction, weak submodularity, iterative algorithms
\end{abstract}

%
\section{Introduction}\label{sec:intro}
Network data that are naturally supported on vertices of a graph are becoming increasingly ubiquitous, 
with examples ranging from the measurements of neural activities in different regions of the brain 
\cite{huang2016graph} to vehicle trajectories over road networks \cite{deri_nyc_taxi}. Predicated on
the assumption that the properties of a network process relate to the underlying graph, the goal of 
graph signal processing (GSP) is to broaden the scope of traditional signal processing tasks and 
develop algorithms that fruitfully exploit this relational structure \cite{shuman2013,sandryhaila2013}.

Consider a network represented by a graph $\mathcal{G}$ consisting of a node set $\mathcal{N}$ of 
cardinality $N$ and a weighted adjacency matrix $\mathbf{A} \in \mathbb{R}^{N \times N}$ whose
$(i,j)$ entry, $\A_{ij}$, denotes weight of the edge connecting node $i$ to node $j$. A \textit{graph 
signal} $\mathbf{x} \in \mathbb{R}^{N}$ is a vertex-valued network process that can be represented 
by a vector of size $N$ supported on $\mathcal{N}$, where its $i\ts{th}$ component denotes the 
signal value at node $i$.

A cornerstone problem in GSP that has drawn considerable attention in recent years pertains to 
sampling and reconstruction of graph signals 
\cite{shomorony2014sampling,tsitsvero2016signals,anis2016efficient,chen2015discrete,
chepuri2016subsampling,marques2016sampling,gama2016rethinking,jayawant2018distance,chamon2017greedy}. The 
task of selecting a subset of nodes whose signals enable reconstruction of the information in the 
entire graph with minimal loss is known to be NP-hard. Conditions for exact reconstruction of 
graph signals from noiseless samples were put forth 
in~\cite{shomorony2014sampling,tsitsvero2016signals,anis2016efficient,chen2015discrete}. 
Existing approaches for sampling and reconstruction of graph signals can be categorized in two 
main groups -- selection sampling \cite{chen2015discrete} and aggregation sampling \cite{marques2016sampling}. The focus of the current paper is on the former.
\subsection{Related work}
Sampling of noise-corrupted signals using randomized schemes including uniform and leverage 
score sampling is studied in \cite{chen2016signal}; there, optimal sampling distributions and
performance bounds are derived. Building on the ideas of variable density sampling from 
compressed sensing, \cite{puy2016random} derives random sampling schemes and proves that 
$\mathcal{O}(k \log k)$ samples are sufficient to recover all $k$-spectrally sparse signals with high 
probability. Moreover, \cite{puy2016random} provides a fast technique for accurate estimation of 
the optimal sampling distribution. Recent work \cite{tremblay2017graph} relies on loop-erased 
random walks on graphs to speed up sampling of bandlimited signals. In 
\cite{chepuri2016subsampling,chamon2017greedy}, reconstruction of graph signals and their 
power spectrum density was studied and schemes based on the greedy sensor selection algorithm 
\cite{shamaiah2010greedy,shamaiah2012greedy} were developed. However, the performance 
guarantees in \cite{chen2016signal,chamon2017greedy} are restricted to the case of stationary 
graph signals, i.e., the covariance matrix in the nodal or spectral domains is required to have a
certain structure (e.g., diagonal; see also 
\cite{marques2016stationaryTSP16,perraudinstationary2016,girault_stationarity}). 

An influential work \cite{pesenson2008sampling} presents a method that enables recovery of some 
bandlimited functions on a simple undirected unweighted graph using signal values observed on the 
so-called uniqueness sets of vertices; see also \cite{narang2013signal} and \cite{anis2014towards}. 
An iterative local set-based algorithm that relies on graph partitioning to improve convergence rate 
of bandlimited graph signals reconstruction is proposed in \cite{wang2015local}.

The sampling approach in \cite{marques2016sampling} relies on collecting observations at a single 
node instead of a subset of nodes via successive applications of the so-called graph shift operator 
and aggregating the results. Specifically, shifted versions of the signal are sampled at a 
single node which, under certain conditions, enables recovery of the signal at all nodes. While the 
aggregation sampling in \cite{marques2016sampling} reduces to the classical sampling of time signals, 
the required inspection of the invertibility of the submatrix of eigenvectors is computationally 
expensive. Moreover, the recovery of graph signals from their partial samples collected via the 
aggregation scheme requires the first $k$ components (signal bandwidth) to be distinct, which may 
not be the case in certain applications. 
Table \ref{hopkins} summarizes properties of a few

A main challenge in sampling and reconstruction of spectrally sparse graph signals is the problem of 
identifying their support  \cite{marques2016sampling,di2016adaptive,romero2017kernel,narang2013signal,anis2016efficient}.
In \cite{anis2014towards,anis2016efficient}, support identification of smooth graph signals is studied. 
However, the techniques in \cite{narang2013signal,anis2016efficient} rely solely on a user-defined 
sampling strategy and the graph Laplacian, and disregard the availability of observations of the graph 
signal. A similar scheme is developed in \cite{marques2016sampling} for aggregation sampling where 
under established assumptions on the topology of a graph, conditions for the exact support identification 
from noiseless measurements are established. In particular, the aggregation sampling method of 
\cite{marques2016sampling} requires twice as many samples as the bandwidth of the graph signal 
(i.e., $k$) to guarantee perfect recovery in the noiseless setting. An alternating minimization approach 
that jointly recovers unknown support of the signal and designs a sampling strategy in an iterative 
fashion is proposed in \cite{di2016adaptive}. However, convergence of the alternating scheme in 
\cite{di2016adaptive} is not guaranteed and the conditions for exact support identification are 
unknown \cite{di2016adaptive}.
\subsection{Contribution}
{\color{black}
Although tremendous efforts have been made to address fundamental theoretical and algorithmic questions in sampling and reconstruction of bandlimited graph signals, the high computational costs of existing methods that deliver competitive reconstruction performance typically render their applicability challenging, especially in applications dealing with large-scale and high-dimensional graphs. Therefore, developing scalable, efficient, and accelerated sampling and reconstruction algorithms with provable performance is highly desired.}

In this paper, we consider the task of sampling and reconstruction of spectrally sparse graph signals in 
various settings.
We first study the non-Bayesian scenario where no prior information about signal covariance 
is available. Based on ideas from compressed sensing, we develop a novel and efficient iterative 
sampling approach that exploits the low-cost selection criterion of the orthogonal matching pursuit 
algorithm \cite{pati1993orthogonal} to recursively select a subset of nodes of the graph. We 
theoretically demonstrate that in the noiseless case the original $k$-spectrally sparse signal can be 
recovered exactly from the set of selected nodes with cardinality $k$. In the case of $\ell_2$-norm 
bounded noise, we establish a bound on the worst-case reconstruction error of the proposed 
algorithm that turns out to be proportional to the bound on the $\ell_2$-norm of the noise term. 
The proposed scheme requires only that the graph adjacency matrix is normal, a typical assumption in 
prior works on the sampling of graph signals. Therefore, the proposed iterative algorithm guarantees 
recovery for a wide class of graph structures.

Next, we study a Bayesian scenario where the graph signal is a non-stationary network process 
with a known non-diagonal covariance matrix.  Following 
\cite{chamon2017greedy,shamaiah2010greedy,shamaiah2012greedy}, we formulate the sampling 
task as the problem of maximizing a monotone weak submodular function that is directly related to 
the mean square error (MSE) of the linear estimator of the original graph signal. To find a sub-optimal 
solution to this combinatorial optimization problem, we propose a randomized-greedy algorithm that 
is significantly faster than the greedy sampling method in 
\cite{chamon2017greedy,shamaiah2010greedy,shamaiah2012greedy}. We theoretically analyze 
performance of the proposed randomized-greedy algorithm and demonstrate that the resulting MSE 
is a constant factor away from the MSE of the optimal sampling set. Unlike the prior work in 
\cite{chamon2017greedy}, our results do not require stationarity of the graph signal. Furthermore, in 
contrast to the existing theoretical works, we do not restrict our study to the case of additive white noise. 
Instead, we assume that the noise coefficients are independent and allow the power of noise to vary 
across individual nodes of the graph. 

Simulation studies on both synthetic and real world graphs verify our theoretical findings 
and illustrate that the proposed sampling framework compares favorably to competing alternatives 
in terms of both accuracy and runtime.

Preliminary results of this work is published in \cite{hashemi2018sampling,hashemi2018sampling_global}. In addition to providing the details of proofs which were missing from \cite{hashemi2018sampling,hashemi2018sampling_global} and discussing the computational complexity of the proposed algorithms, for the Bayesian setting, we extend the scope of our study to provide {\it high probability} error bounds for the achievable mean-square error performance of the proposed randomized greedy sampling schemes. Finally, in our extensive experimental study, we discuss two new applications of graph sampling, namely, localization of UAVs under power constraints and semi-supervised face clustering via subspace learning. We further demonstrate the efficacy of our randomized greedy algorithm on a large-scale preferential attachment graph with 10,000 nodes.
\subsection{Organization}
The rest of the paper is organized as follows. Section \ref{sec:pre} reviews the relevant background 
and concepts. In Section \ref{sec:bl}, we formally state the sampling problem and develop the proposed 
iterative selection sampling method. In Section \ref{sec:alg}, we study the Bayesian setting, introduce 
the randomized-greedy algorithm for the sampling task and theoretically analyze its performance. 
Section \ref{sec:sim} presents simulation results while the concluding remarks are stated in Section 
\ref{sec:concl}. 

\begin{table*}[t]\centering
\footnotesize
	\caption{Properties of sampling schemes for spectrally sparse signals in scenarios where the basis matrix $\U$ is known.}
	\ra{1}
	\begin{tabular*}{1\linewidth}{@{}ccccc@{}}\toprule
		Assumption &\phantom{}&Optimality criteria&\phantom{}&Algorithms\\ \midrule
		noise-free samples, non-Bayesian&&full rank $\U_\S$ &&\begin{tabular}{@{}c@{}}Gaussian elimination, greedy\cite{shomorony2014sampling},\\random\cite{puy2016random,chen2016signal} \end{tabular}  \\\midrule	
		noisy samples, non-Bayesian&& $\min \mathrm{Tr}(\E[(\x-\hat{\x})((\x-\hat{\x})^\top])$&&\begin{tabular}{@{}c@{}}Gaussian elimination, greedy\cite{chen2015discrete},\\random\cite{puy2016random,chen2016signal} \end{tabular} \\\midrule	
		noisy samples, Bayesian&& $\min \mathrm{Tr}(\E[(\x-\hat{\x})((\x-\hat{\x})^\top])$&& greedy\cite{chamon2017greedy}, convex optimization\cite{joshi2009sensor}\\
		\bottomrule
	\end{tabular*} 
	\label{hopkins}
\end{table*}

\section{Preliminaries}\label{sec:pre}
In this section, we overview notation, concepts, and definitions that are used in the development of the proposed algorithmic and theoretical frameworks.
\subsection{Notations}
Bold capital letters denote matrices while 
bold lowercase letters represent vectors. Sets are denoted by calligraphic letters and $|\S|$ denotes the 
cardinality of set $\S$. $\A_{ij}$ denotes the $(i,j)$ entry of $\A$, $\a_j$ 
($\mathbf{a}^{j}$) is the $j\ts{th}$ row (column) of $\A$, $\A_{\S,r}$ ($\A_{\S,c}$) is the submatrix of $\A$ 
that contains rows (columns) indexed by the set $\S$, and $\lambda_{max}(\A)$ and $\lambda_{min}(\A)$ 
are the largest and smallest eigenvalues of $\A$, respectively. $\mathbf{P}_\S^\bot=\I_n-\A_{\S,r}^\top 
(\A_{\S,r}^\top)^\dagger$ is the projection operator onto the orthogonal complement of the subspace spanned 
by the rows of $\A_{\S,r}$, where $\A^\dagger=\left(\A^{\top}\A\right)^{-1}\A^{\top}$ denotes the 
Moore-Penrose pseudo-inverse of $\A$ and $\I_n \in \R^{n\times n}$ is the identity matrix. Finally, 
$\text{supp}(\x)$ returns the support of $\x$ and $[n] := \{1,2,\dots,n\}$.
\subsection{Spectrally sparse graph signals}
Let $\x$ be a graph signal which is $k$-spectrally sparse in a given basis $\mathbf{V} \in \mathbb{R}^{N \times N}$. This means that the signal's so-called graph Fourier transform (GFT) $\bar{\x} = \V^{-1} \x$ is $k$-sparse. There are several choices for $\mathbf{V}$ in literature with most aiming to decompose a graph signal into different modes of variation with respect to the graph topology. For instance, $\mathbf{V} = [\mathbf{v}_{1},\cdots,\mathbf{v}_{N}]$ can be defined via the Jordan decomposition of the adjacency matrix \cite{DSP_freq_analysis,deri2017spectral}, through the eigenvectors of the Laplacian when $\mathcal{G}$ is undirected \cite{shuman2013}, or it can be obtained as the result of an optimization procedure \cite{shafipour2017digraph,sardellitti}. In this paper, we assume that the adjacency matrix $\A = \V \mathbf{\Lambda}\V^{-1}$ is normal which in turn implies $\V$ is unitary and $\V^{-1} = \V^\top$.

Recall that since $\x$ is spectrally sparse, $\bar{\x}$ is sparse with at most $k$ nonzero entries. Let $\K$ 
be the support set of $\bar{\x}$, where $|\K| = k$. Then one can write $\x = \U\bar{\x}_\K$, where 
$\U = \V_{\K,c}$. In the sequel, without loss of generality we assume $\U$ does not contain all-zero 
rows; otherwise, one could omit the all-zero rows of $\U$ and their corresponding nodes from the graph 
as they provide no meaningful information about the graph signals. Moreover, we proceed by assuming that 
the support set $\K$ is known. 

\textbf{Remark 1.} As in the prior work on sampling graph signals 
\cite{puy2016random,chen2016signal,chen2015discrete,chamon2017greedy,marques2016sampling,gama2016rethinking}, 
our proposed schemes require the graph Fourier transform (GFT) bases (i.e., $\V$) as input; this involves eigenvalue 
decomposition of $\A$ which may be computationally intensive for large graphs. The focus of this paper, however, is not on 
the pre-processing step of finding $\V$ but rather on developing efficient sampling algorithms with theoretical performance 
guarantees on the achievable reconstruction error in a variety of settings.
\subsection{Submodularity and weak submodular functions}
An important concept in contemporary combinatorial optimization is the notion of submodular functions that has recently found applications in many signal processing tasks. Relevant concepts are formally defined below. 
\begin{definition}[\textbf{Submodularity and monotonicity}]
	\label{def:submod}
	Let $\X$ be a ground set. Set function $f:2^\X\rightarrow \mathbb{R}$ is submodular if 
	\begin{equation*}
	f(\S\cup \{j\})-f(\S) \geq f(\T\cup \{j\})-f(\T)
	\end{equation*}
	for all subsets $\S\subseteq \T\subset \X$ and $j\in \X\backslash \T$. The term $f_j(\S):=f(\S\cup \{j\})-f(\S)$ is the marginal value of adding element $j$ to set $\S$. Furthermore, $f$ is monotone if $f(\S)\leq f(\T)$ for all $\S\subseteq \T\subseteq \X$.
\end{definition}
In many applications, the objective function of a combinatorial optimization problem of interest is not submodular. The notion of set functions with bounded curvature captures these scenarios by generalizing  the concept of submodularity.
\begin{definition}[\textbf{Curvature}]
	The maximum element-wise curvature of a monotone non-decreasing function $f$ is defined as
	\begin{equation*}
	{\cal C}_{f}=\max_{l\in [N-1]}{\max_{(\S,\T,i)\in \mathcal{\X}_l}{f_i(\T)\slash f_i(\S)}},
	\end{equation*}
	where $\mathcal{\X}_l = \{(\S,\T,i)|\S \subset \T \subset \X, i\in \X \backslash \T, |\T\backslash \S|=l,|\X|=N\}$. 
\end{definition}
The maximum element-wise curvature essentially quantifies  how close the set function is to being submodular. It is worth noting that a set function $f(\S)$ is submodular if and only if its maximum element-wise curvature satisfies 
${\cal C}_{f} \le 1$. When ${\cal C}_{f} > 1$, $f(\S)$ is called a \textit{weak submodular} function.
\section{Sampling of Spectrally Sparse Graph Signals}\label{sec:bl}
In this section, we study the problem of sampling spectrally sparse signals with known support. In 
particular, we assume that a graph signal $\x$ is sparse given a basis $\V$ and that
$\A = \V \mathbf{\Lambda}\V^\top$, where $\A$ is the adjacency matrix of the undirected graph 
$\mathcal{G}$; alternatively, we may use the Laplacian matrix $\mathbf{L}$ to characterize the undirected
graph. We can also consider any orthogonal basis for general directed graphs; see e.g., \cite{shafipour2017digraph}. We first consider the noise-free scenario (Section \ref{SS:noiseless}) and then extend our results to the case of 
sampling and reconstruction from noisy measurements (Section \ref{SS:noisy}).
\subsection{Sampling strategy} \label{SS:noiseless}
In selection sampling (see, e.g.\cite{chen2015discrete}), sampling a graph signal amounts to 
finding a matrix $\C \in \{0,1\}^{k\times N}$ such that $\tilde{\x} = \C\x$, where $\tilde{\x}$ 
denotes the sampled graph signal. Since $\x$ is spectrally sparse with support $\mathcal{K}$ and 
$\x = \U\bar{\x}_\K$, it holds that $\tilde{\x} = \C\U\bar{\x}_\K.$
The original signal can then be reconstructed as
\begin{equation}\label{eq:rec1}
\hat{\x} =\U\bar{\x}_\K= \U(\C\U)^{-1}\tilde{\x}.
\end{equation}
According to \eqref{eq:rec1}, a necessary and sufficient condition for perfect reconstruction (i.e., 
$\hat{\x} = \x$) from noiseless observations is guaranteed by the invertibility of matrix $\C\U$. However, 
as argued in \cite{marques2016sampling,shomorony2014sampling} (see, e.g. Section III-A in 
\cite{marques2016sampling}), current random selection sampling approaches cannot construct 
a sampling matrix to ensure $\C\U$ is invertible for an arbitrary graph; moreover, invertibility of 
$\C\U$ is checked by inspection which for large graphs requires intensive computational effort.
To overcome these issues, motivated by the well-known OMP algorithm in compressed sensing 
\cite{pati1993orthogonal}, we propose a simple iterative scheme with complexity 
$\mathcal{O}(Nk^2)$ that guarantees perfect recovery of $\x$ from the 
sampled signal $\tilde{\x}$.

The proposed approach (see Algorithm \ref{alg:1}) works as follows. First, the algorithm chooses 
a node of the graph with index $\ell$ as a {\it residual node} such that $\ell  = \argmin_{j\in [N]} \|\u_j\| $. \footnote{Intuitively, this node has weaker representative power compared to other point and since in Algorithm 1 the residual node is excluded from the selection procedure, this choice empirically leads to smaller reconstruction error in noisy scenario.} Then, 
in the $i\ts{th}$ iteration the algorithm identifies a node -- excluding the residual node --  with index $s_j$ to be included in the 
sampling set $\S$ according to 
\begin{equation}\label{eq:gcss}
s_j = \argmax_{j \in {\mathcal{N}\backslash{\ell}} \backslash\S}\frac{|\r_{i-1}^\top \u_j|^2}{\|\u_j\|_2^2},
\end{equation}
where $\r_i = \mathbf{P}_\S^\bot\u_\ell$ is a residual vector initialized as $\r_0 = \u_\ell$, and $\mathbf{P}_\S^\bot=\I_n-\U_{\S,r}^\top 
(\U_{\S,r}^\top)^\dagger$. This procedure is 
repeated for $k$ iterations to construct the sampling set $\S$.

\textbf{Remark 2.} Optimization \eqref{eq:gcss} is related to the greedy column subset selection 
approach in \cite{farahat2015greedy}. Specifically, both methods attempt to identify a subset of the rows/columns 
 that best represent the entire matrix. However, they focus on different applications which in turn 
results in different definitions of the residuals. In \cite{farahat2015greedy}, the residual is defined as the original 
matrix itself. Hence the computational complexity of the greedy approach in \cite{farahat2015greedy} is 
significantly higher than that of Algorithm 1 where the residual is merely a vector.

Theorem \ref{thm:iss} demonstrates that Algorithm \ref{alg:1} returns a sampling set which ensures 
perfect recovery of the graph signal $\x$ in the noise-free scenario.
\begin{theorem}\label{thm:iss}
Let $\S$ denote the sampling set constructed by Algorithm \ref{alg:1} and let $\C$ be the corresponding 
sampling matrix such that $|\S| = k$. Then, matrix $\C\U$ is always invertible.
\end{theorem}
\begin{proof}
See Appendix \ref{pf:iss}.
\end{proof}
Theorem \ref{thm:iss} states that as long as the adjacency matrix $\A$ is normal, the proposed selection 
scheme guarantees perfect reconstruction of the original signal from its noiseless samples.  
Therefore,  in contrast to existing random selection sampling and aggregation sampling schemes \cite{chen2015discrete,puy2016random,marques2016sampling} that require strong conditions on $\A$ 
(e.g., eigenvalues of $\A$ to be distinct), Algorithm \ref{alg:1} guarantees recovery for a wider class of 
graphs.
\subsection{Complexity analysis}
The worst-case computational complexity of Algorithm \ref{alg:1} is analyzed next. In the $i\ts{th}$ iteration, 
step 6 costs $\mathcal{O}(k(N-i))$ as one needs to search over $N-i$ rows of $\U$ and compute inner-products 
of $k$-dimensional vectors in order to evaluate the selection criterion. Step 8 is a matrix-vector product whose
complexity is $\mathcal{O}(k^2)$. Note that in our implementation we use the modified Gram-Schmidt (MGS) 
algorithm to update the residual vector with a significantly lower complexity of $\mathcal{O}(ki)$. Thus, the total 
cost of the $i\ts{th}$ iteration is $\mathcal{O}(k(N-i)+ki) = \mathcal{O}(k(N-i))$. Since $i\leq k$ and there are $k$ iterations, the overall complexity of Algorithm \ref{alg:1} is $\mathcal{O}(Nk^2)$. {\color{black}Please refer to Table \ref{table2} for a comparison between computational costs of proposed schemes in this paper to the existing methods.}
\renewcommand\algorithmicdo{}	
\begin{algorithm}[t]
	\caption{Iterative Selection Sampling}
	\label{alg:1}
	\begin{algorithmic}[1]
		\STATE \textbf{Input:} $\U$, $k$, number of samples $m\geq k$.
		\STATE \textbf{Output:} Subset $S\subseteq \mathcal{N} $ with $|\S|=m$.
		\STATE Initialize $\S =  \emptyset$, $\r_0 = \u_\ell$ for $\ell  = \argmin_{j\in [N]} \|\u_j\| $, and $i = 0$.
		\WHILE{$|\S|<m$}\vspace{0.2cm}
		\STATE $i\leftarrow i+1$\vspace{0.2cm}
		\STATE $s_i = \argmax_{j \in {\mathcal{N}\backslash\{\ell\}} \backslash\S}\frac{|\r_{i-1}^\top \u_j|^2}{\|\u_j\|_2^2}$\vspace{0.2cm}
		\STATE Set $\S \leftarrow \S\cup \{s_i\}$\vspace{0.2cm}
		\STATE $\r_i = \mathbf{P}_\S^\bot\u_\ell$\vspace{0.2cm}
		\ENDWHILE
		\RETURN $\S$.
	\end{algorithmic}
\end{algorithm}
\subsection{Sampling in the presence of noise} \label{SS:noisy}
Here we provide an extension of the proposed selection sampling scheme to the scenarios where only noisy 
observations of the graph nodes are available. Note that due to noise, perfect reconstruction is no 
longer possible. Nevertheless, we provide an upper bound on the reconstruction error of the proposed 
sampling scheme as a function of the noise covariance and the sampling matrix $\C$. Another distinguishing
aspect of sampling and reconstruction in the presence of noise is that, to achieve better reconstruction 
accuracy, it may be desirable to select $m\geq k$ nodes as the sampling set. This stands in contrast to 
the noiseless case where, as we proved, $m= k$ sampling nodes are sufficient for perfect reconstruction 
if the sampling set is constructed by Algorithm \ref{alg:1}.

Let $\y = \x + \n$ be the noise-corrupted signal, where $\n \in \R^{N}$ denotes the zero-mean 
noise vector with covariance matrix $\E[\n\n^\top]=\Q$. We also assume that the support $\K$ 
is known. Therefore, since $\x = \U\bar{\x}_\K$, the samples $\tilde{\x}$ and the non-zero 
frequency components of $\x$ are related via the linear model
\begin{equation}\label{eq:modeln}
\tilde{\x} = \y_\S = \U_{\S,r}\bar{\x}_\K +\n_{\S},
\end{equation}
where $\U_{\S,r} = \C\U$, $\y_\S = \C\y$, and  $\n_\S = \C\n$. The reconstructed signal in the Fourier domain is found by seeking the least square solution and satisfies the normal equation \cite{kay1993fundamentals},
\begin{equation}\label{eq:normal}
\U_{\S,r}^\top\Q_\S^{-1}\U_{\S,r}\hat{\bar{\x}} = \U_{\S,r}^\top \Q_\S^{-1} \tilde{\x},
\end{equation}
where $\Q_\S = \C\Q\C^\top$ is the covariance of $\n_\S$.

 If $\U_{\S,r}^\top\Q_\S^{-1}\U_{\S,r}$ is invertible, we can recover the original graph signal up to an error term as stated~in~the~following~proposition.
\begin{proposition}\label{thm:nois}
Let $\S$ be the sampling set constructed by Algorithm \ref{alg:1} and let $\C$ be the corresponding sampling 
matrix. Moreover, let us denote $\U_{\S,r} = \C\U$. Then, with probability one the matrix 
$\U_{\S,r}^\top\Q_\S^{-1}\U_{\S,r}$ is invertible. Furthermore, if $\|\n\|_2 \leq \epsilon_\n$, the reconstruction 
error of the signal reconstructed from $\S$ satisfies
\begin{equation}\label{eq:boundn}
\|\hat{\x}-\x\|_2 \leq \sigma_{\max}((\U_{\S,r}^\top\Q_\S^{-1}\U_{\S,r})^{-1}\U_{\S,r}^\top \Q_\S^{-1})\epsilon_\n,
\end{equation}
where $\sigma_{\max}(.)$ outputs the maximum singular value of its matrix argument.
\end{proposition}
\begin{proof}
See Appendix \ref{pf:thmnois}.
 \end{proof}

Compared to the noiseless scenario where the main challenge is to ensure that
$\C\U$ is invertible, in the presence of noise we are interested in finding a sampling scheme with 
the lowest reconstruction error. Although Proposition \ref{thm:nois} provides a performance bound 
for any sampling matrix $\C$ constructed by Algorithm \ref{alg:1}, our specific choice of the residual 
node, $\ell  = \argmin_{j\in [N]} \u_j$, is not exploited in the proof of Proposition \ref{thm:nois} 
and further analysis along those lines is left as part of the future work. We empirically 
observed that with the proposed choice of the residual, the matrix product on the right-hand side of 
\eqref{eq:boundn} has smaller maximum singular value than if the residual node is selected uniformly 
at random. We also note that the statistics of noise is not exploited when constructing $\C$. This is
similar to state-of-the-art random selection sampling and aggregation sampling schemes 
\cite{chen2015discrete,puy2016random,marques2016sampling} where one needs to rely on 
exhaustive search over the space of all sampling matrices to find the one that results in the lowest 
MSE. In the Bayesian setting studied in Section \ref{sec:alg} where one assumes a prior distribution 
on $\x$, the original signal can be reconstructed up to an error term for any $\C \in \R^{m\times N}$ 
with $m\geq k$. Therefore, invertibility of $\C\U$ is not a concern in the Bayesian case where we 
focus on the construction of a sampling set $\S$ with the lowest~reconstruction~error.

Note that the Gaussian elimination scheme also finds a full-rank submatrix $\U_\S$. According to 
\cite{anis2016efficient}, the sampling set found by Gaussian elimination with partial row pivoting 
corresponds to indices of the pivot rows. Therefore, in contrast to Algorithm 1 that takes into 
accounts representative power of each node in all frequency components (by considering the 
$\ell_2$ norm of $\u_j$'s and their correlation with the residual), Gaussian elimination with partial 
row pivoting only considers individual frequency components when forming the sampling set. 
Hence, the signal reconstructed by such a scheme may not be robust to noise statistics. On the 
other hand, by choosing the residual node according to $\ell  = \argmin_{j\in [N]} \u_j$, 
Algorithm 1 finds an invertible submatrix and further finds a subset of rows of $\U$ with 
strong representation capability.
\section{Bayesian Sampling of Graph Signals}\label{sec:alg}

\begin{table}[t]\centering
	{\color{black}
		\caption{Computational complexity comparison between the proposed algorithms and the existing methods.}
		\ra{1}
		\begin{tabular*}{0.75\linewidth}{@{}ccccccc@{}}\toprule
			Algorithm &\phantom{1111}&Setting&\phantom{1111}&Complexity\\ \midrule
			Proposed Algorithm 1&& non-Bayesian&&$\mathcal{O}(k^2N)$ \\\midrule	
			Greedy \cite{chen2015discrete}&& non-Bayesian &&$\mathcal{O}(k^4N)$\\\midrule	
			Greedy \cite{chamon2017greedy}&& Bayesian &&${\cal O}(Nk^3)$\\\midrule		
			Proposed Algorithm 2&& Bayesian&&${\cal O}(Nk^2)$\\
			\bottomrule
		\end{tabular*} 
	\label{table2}
	}
\end{table}
So far we have considered the problem of sampling in scenarios where the graph signal is not stochastic. In this 
section, we consider the problem of sampling and interpolation in a Bayesian setting where the graph signal is a 
non-stationary network process. To this end, we adopt the following definition of stationarity, recently proposed in \cite{perraudinstationary2016}.
\begin{definition}
	A stochastic graph signal $\x$ is graph wide-sense stationary (GWSS) if and only if the matrix
	\begin{equation}
	\E[\bar{\x}\bar{\x}^\top] = \mathbf{V}^\top\E[{\x}{\x}^\top]\mathbf{V}
	\end{equation}
	is diagonal.
\end{definition}

In addition to our novel algorithmic contributions, the setting we consider in this section is more general than those 
considered in \cite{chamon2017greedy,ma,chamon2017approximate,chamon2017mean}. Specifically, unlike
the prior work \cite{chamon2017greedy}, we assume that the signal in not necessarily stationary with respect to 
$\mathcal{G}$ and that $\bar{\x}$ is a zero-mean random vector with generally non-diagonal covariance matrix 
$\E[\bar{\x}\bar{\x}^\top]$ = $\P$. Furthermore, we do not restrict our study to the case of additive white noise. Rather, 
we consider a more practical setting where the noise terms are independent but the noise power varies across
individual nodes of the graph. That is, if $\y = \x + \n$ denotes the noise-corrupted signal, $\n \in \R^{N}$ is a 
zero-mean noise vector with covariance matrix $\E[\n\n^\top]=\Q =  \mathrm{diag}(\sigma_1^2,\dots,\sigma_N^2)$. 
Note that this particular scenario is not explored in \cite{chamon2017greedy} or the related sensor selection 
and experimental design schemes \cite{ma,chamon2017approximate,chamon2017mean}.

Let $\S$ denote a sampling set of $m\geq k$ graph nodes. Since $\x = \U\bar{\x}_\K$, the samples $\y_\S$ and the 
non-zero frequency components of $\x$ are related via the Bayesian linear model
\begin{equation}\label{eq:model}
\y_\S = \U_{\S,r}\bar{\x}_\K +\n_{\S}.
\end{equation}
As before, in order to find $\hat{\x}$ it suffices to estimate $\bar{\x}_\K$ based on $\y_\S$. The least mean-square 
estimator of $\bar{\x}_\K$, denoted by $\hat{\bar{\x}}_{\K}$, is the Bayesian counterparts of the normal equations 
in the Gauss-Markov theorem (see, e.g.  \cite[Ch. 10]{kay1993fundamentals}). In other words, it is given by
\begin{equation}\label{eq:estf}
\hat{\bar{\x}}_\K = \bar{\Sb}_\S \U_{\S,r}^\top\Q_\S^{-1} \y_\S,
\end{equation}
where 
\begin{equation}\label{eq:covf}
\begin{aligned}
\bar{\Sb}_\S &= \left(\P^{-1}+\U_{\S,r}^\top\Q_\S^{-1}\U_{\S,r}\right)^{-1} \\
&= \left(\P^{-1}+\sum_{j\in\S}\frac{1}{\sigma^2_j}\u_j\u_j^\top\right)^{-1}
\end{aligned}
\end{equation}
is the error covariance matrix of $\hat{\bar{\x}}_{\K}$. Therefore, $\hat{\x} = \U\hat{\bar{\x}}_{\K}$ and its error covariance matrix is ${\Sb}_\S = \U \bar{\Sb}_\S \U^\top$.

The problem of sampling for near-optimal reconstruction can now be formulated as the task of choosing $\S$ so as to minimize the MSE of the estimator $\hat{\x}$. Since the MSE is defined as the trace of the error covariance matrix, we arrive at the following optimization problem,
\begin{equation}\label{eq:probs}
\begin{aligned}
& \underset{\S}{\text{min}}
\quad \mathrm{Tr}\left({\mathbf{\Sigma}}_\S\right)
& \text{s.t.}\quad \S \subseteq \mathcal{N}, \phantom{k}|\S|\leq m.
\end{aligned}
\end{equation}
Using trace properties and the fact that $\U^\top\U = \I_m$, \eqref{eq:probs} simplifies to
\begin{equation}\label{eq:probf}
\begin{aligned}
& \underset{\S}{\text{min}}
\quad \mathrm{Tr}\left(\mathbf{\bar{\Sigma}}_\S\right)
& \text{s.t.}\quad \S \subseteq \mathcal{N}, \phantom{k}|\S| \leq m.
\end{aligned}
\end{equation}
The optimization problem \eqref{eq:probf} is NP-hard and evaluating all $\genfrac(){0pt}{2}{N}{m}$ possibilities to 
find the exact solution is intractable even for relatively small graphs. To this end, we propose an alternative 
to find a near-optimal solution in polynomial time. In \cite{chamon2017greedy}, similar to the greedy sensor 
selection approach of \cite{shamaiah2010greedy,shamaiah2012greedy}, a greedy algorithm is proposed for the 
described Bayesian setting and its performance is analyzed under the assumption that the graph signal is stationary 
and the noise is white. In applications dealing with extremely large graphs, the greedy algorithm in 
\cite{chamon2017greedy} might be computationally infeasible. Moreover, 
the graph signal is not necessary stationary and, perhaps more importantly, different nodes of a graph may experience 
different levels of noise. To address these challenges, motivated by the algorithm recently developed in 
\cite{mirzasoleiman2014lazier} for maximization of {\it strictly} submodular functions, we develop a randomized-greedy 
algorithm for Bayesian sampling of graph signals that is significantly faster than the greedy algorithm. In addition, 
by leveraging the notion of weak submodularity, we establish performance bounds for the general setting of 
non-stationary graph signals. 
\subsection{Randomized-greedy selection sampling}
Following \cite{shamaiah2010greedy,shamaiah2012greedy,chamon2017greedy}, we start by formulating 
\eqref{eq:probf} as a  set function maximization task. Let $f(\S) = \mathrm{Tr}(\P-\bar{\Sb}_\S)$. Then
\eqref{eq:probf} can equivalently be written as
\begin{equation}\label{eq:probsub}
\begin{aligned}
& \underset{\S}{\text{max}}
\quad f(\S)
& \text{s.t.}\quad \S\subseteq\mathcal{N},\quad |\S| \leq m.
\end{aligned}
\end{equation}

In Proposition \oldref{thm:p} below, by applying the matrix inversion lemma 
\cite{bellman1997introduction} we establish that $f(\S)$ is monotone and weakly submodular. 
Moreover, we derive an efficient recursion to find the marginal gain of adding a new node 
to the sampling set $\S$. 
\begin{proposition}\label{thm:p}
\textit{$f(\S) = \mathrm{Tr}(\P-\bar{\Sb}_\S)$ is a weak submodular, monotonically increasing set function, $f(\emptyset)=0$, and for all $j \in \mathcal{N}\backslash \S$
\begin{equation}\label{eq:mg}
f(\S\cup \{j\})-f(\S) = \frac{\u_j^\top\bar{\Sb}_S^{2}\u_j}{\sigma_j^{2}+\u_j^\top\bar{\Sb}_\S\u_j},\: \text{ and }
\end{equation}
\begin{equation}\label{eq:upf}
\bar{\Sb}_{\S \cup\{j\}} = \bar{\Sb}_\S-\frac{\bar{\Sb}_{\S}\u_{j}\u_{j}^\top\bar{\Sb}_{\S}}{\sigma_j^2+\u_{j}^\top\bar{\Sb}_{\S}\u_{j}}.
\end{equation}}
\end{proposition}
\begin{proof}
See Appendix \ref{pf:mono}.
\end{proof}
Proposition \ref{thm:p} enables efficient construction of the sampling set in an iterative fashion. To further reduce 
the computational cost, we propose a randomized-greedy algorithm for selection sampling with minimal MSE that selects a sampling set in the following way. 
Starting with $\S = \emptyset$, at iteration $(i+1)$ of the algorithm, a subset $\mathcal{R}$ of size $s$ is sampled 
uniformly at random and without replacement from $\mathcal{N} \backslash \S$. The marginal gain of each node in 
$\mathcal{R}$ is found using \eqref{eq:mg}, and the one corresponding to the highest marginal gain is added to $\S$. 
Then, the algorithm employs \eqref{eq:upf} to update $\bar{\Sb}_\S$ for the subsequent iteration. This procedure is 
repeated until some stopping criteria, e.g., a condition on the cardinality of $\S$ is met. Regarding $s$, we follow the 
suggestion in \cite{mirzasoleiman2014lazier} and set $s=\frac{N}{m}\log\frac{1}{\epsilon}$ where 
$e^{-m}\leq \epsilon<1$ is a predetermined parameter that controls the trade-off between the computational cost 
and MSE of the reconstructed signal; randomized-greedy algorithm with smaller $\epsilon$ produces sampling solutions 
with lower MSE while the one with larger $\epsilon$ requires lower computational cost. Note that if $\epsilon = e^{-m}$, 
the randomized-greedy algorithm in each iteration considers all the available nodes and hence matches the greedy 
scheme in \cite{chamon2017greedy}. However, as we illustrate 
in our simulation studies, the proposed randomized-greedy algorithm is significantly faster than the greedy method in \cite{chamon2017greedy} for large $\epsilon$ while returning essentially 
the same sampling solution. The randomized-greedy algorithm is formalized as Algorithm \ref{alg:greedy}.
\subsection{Complexity analysis}
To take a closer look at computational complexity of Algorithm \ref{alg:greedy}, note that step 6 costs $\mathcal{O}(\frac{N}{m}k^2\log(\frac{1}{\epsilon}))$ since one needs to compute 
$\frac{N}{m}\log(\frac{1}{\epsilon})$ marginal gains, each requiring ${\cal O}(k^2)$ operations. 
Furthermore, step 7 requires ${\cal O}(k^2)$ arithmetic operations. Since there are $m$ such 
iterations, running time of Algorithm \ref{alg:greedy} is ${\cal O}(Nk^2\log(\frac{1}{\epsilon}))$. {\color{black}Please refer to Table \ref{table2} for a comparison between computational costs of proposed schemes in this paper to the existing methods.}
\renewcommand\algorithmicdo{}	
\begin{algorithm}[t]
\caption{Randomized-greedy Graph Sampling}
\label{alg:greedy}
\begin{algorithmic}[1]
    \STATE \textbf{Input:}  $\P$, $\U$, $m$, $\epsilon$.
    \STATE \textbf{Output:} Subset $S\subseteq \mathcal{N} $ with $|S|=m \geq k$.
    \STATE Initialize $\S =  \emptyset$, $\bar{\Sb}_{\S}=\P$.
	\WHILE{$|S|<m$}\vspace{0.2cm}
			\STATE Choose $\mathcal{R}$ by sampling $s=\frac{N}{m}\log{(1/\epsilon)}$ indices uniformly at random from $\mathcal{N}\backslash \S$\vspace{0.2cm}
            \STATE $j_s = \argmax_{j\in \mathcal{R}} \frac{\u_j^\top\bar{\Sb}_\S^{2}\u_j}{\sigma_j^{2}+\u_j^\top\bar{\Sb}_\S\u_j}$\vspace{0.2cm}
            \STATE $\bar{\Sb}_{\S \cup\{j_s\}} = \bar{\Sb}_\S-\frac{\bar{\Sb}_{\S}\u_{j}\u_{j}^\top\bar{\Sb}_{\S}}{\sigma_j^2+\u_{j}^\top\bar{\Sb}_{\S}\u_{j}}$\vspace{0.2cm}
            \STATE Set $\S \leftarrow \S\cup \{j_s\}$\vspace{0.2cm}
	\ENDWHILE
	\RETURN $\S$.
\end{algorithmic}
\end{algorithm}
\subsection{Theoretical analysis}\label{sec:anal}
In this section, we analyze performance of the proposed randomized-greedy algorithm in a range of scenarios.

Theorem \ref{thm:exp} below states that if $f(\S)$ is characterized by a bounded maximum element-wise curvature, Algorithm \ref{alg:greedy} returns a sampling subset yielding an MSE that is on average within a multiplicative factor of the MSE associated with the optimal sampling set.
\begin{theorem}\label{thm:exp}
\textit{Let $\mathcal{C}_{f}$ denote the maximum element-wise curvature of $f(\S) = \mathrm{Tr}(\P-\bar{\Sb}_\S)$, 
the objective function in \eqref{eq:probsub}. Let $\alpha =(1-e^{-\frac{1}{c}}-\frac{\epsilon^\beta}{c})$, 
where $c=\max\{1,{\cal C}_{f}\}$, $e^{-m}\leq\epsilon<1$, and $\beta = 1+\max\{0,\frac{s}{2N}-\frac{1}{2(N-s)}\}$. 
Let $\S_{rg}$ be the sampling set returned by the randomized greedy algorithm and let $\O$ denote the optimal 
solution of \eqref{eq:probf}. Then}
\begin{equation}\label{eq:expbound}
\E\left[\mathrm{Tr}(\bar{\Sb}_{\\S_{rg}})\right]\leq \alpha \mathrm{Tr}(\bar{\Sb}_{\O}) + (1-\alpha) \mathrm{Tr}(\P).
\end{equation}
\end{theorem}
\begin{proof}
	The proof of Theorem \ref{thm:exp} relies on the argument that if $s = \frac{N}{m}\log\frac{1}{\epsilon}$, then with high probability the random set $\mathcal{R}$ in each iteration of Algorithm \ref{alg:greedy} contains at least one node from $\O$. See Appendix \ref{pf:exp} for the complete proof.
\end{proof}

Compared to the results of \cite{mirzasoleiman2014lazier} where the maximization of {\it strictly} submodular and 
monotone functions is considered, Theorem \ref{thm:exp} relaxes this assumption and states that 
submodularity is not required for near-optimal performance of the randomized greedy algorithm. In particular, if the set 
function is {\it weak submodular}, Algorithm \ref{alg:greedy} still selects a sampling set with an MSE near that 
achieved by the optimal sampling set. In addition, even if the function is submodular (e.g., when the objective is $\log\det(.)$ 
function instead of the MSE), the approximation factor in Theorem \ref{thm:exp} is tighter than that of 
\cite{mirzasoleiman2014lazier} as the result of the analysis presented in the proof of Theorem \ref{thm:exp}. 
Moreover, a major assumption in \cite{mirzasoleiman2014lazier} is that $\mathcal{R}$ is constructed by sampling 
with replacement. 
In contrast, we assume $\mathcal{R}$ is constructed by sampling without replacement and carry out the 
analysis in this setting.

Next, we study the performance of the randomized greedy algorithm using the tools of probably approximately 
correct (PAC) learning theory \cite{valiant1984theory,valiant2013probably}. That is, not only the sampling set 
selected by Algorithm \ref{alg:greedy} is on expectation near optimal, but the MSE associated with the selected 
sampling set is with high probability close to the smallest achievable MSE. The randomization of Algorithm 
\ref{alg:greedy} can be interpreted as an approximation of the marginal gains of the nodes selected by the 
greedy scheme \cite{chamon2017greedy,shamaiah2010greedy,shamaiah2012greedy}. More specifically, following this interpretation for the 
$i\ts{th}$ iteration we have $f_{j_{rg}}(\S_{rg}) := \eta_i f_{j_{g}}(\S_{g})$, where subscripts $rg$ and $g$ 
indicate the sampling sets and nodes selected by the randomized greedy (Algorithm \ref{alg:greedy}) and the
greedy algorithm in \cite{chamon2017greedy}, respectively, and 
$0<\eta_i\leq 1$ for all $i \in [m]$ are random variables. Following this argument and by employing the Bernstein 
inequality \cite{tropp2015introduction}, we arrive Theorem \ref{thm:pac} which states that the randomized greedy 
algorithm selects a near-optimal sampling set with high probability.   
\begin{theorem}\label{thm:pac}
\textit{Instate the notation and hypotheses of Theorem \ref{thm:exp}. Assume $\{\eta_i\}_{i=1}^m$ is a collection of random variables such that $\E[\eta_i]\geq \mu_{\epsilon}$, for all $i\in[m]$. Then, it holds that
\begin{equation}
\mathrm{Tr}(\bar{\Sb}_{\S_{rg}})\leq 	\left(1- e^{-\sum_{i=1}^m\frac{\eta_i}{mc}}\right) \mathrm{Tr}(\bar{\Sb}_{\O}) + e^{-\sum_{i=1}^m\frac{\eta_i}{mc}} \mathrm{Tr}(\P).
\end{equation}
Moreover, if $\{\eta_i\}_{i=1}^m$ are independent, for all $0<q<1$ with probability at least $1-e^{-Cm}$ it holds 
that}
\begin{equation}\label{eq:pacbound}
\mathrm{Tr}(\bar{\Sb}_{\S_{rg}})\leq \left(1-e^{-\frac{(1-q)\mu_{\epsilon}}{c}}\right) \mathrm{Tr}(\bar{\Sb}_{O}) + e^{-\frac{(1-q)\mu_{\epsilon}}{c}} \mathrm{Tr}(\P)
\end{equation}
for some $C>0$.
\end{theorem}
\begin{proof}
	See Appendix \ref{pf:pac}.
\end{proof}
In our simulation studies (see Section \ref{sec:sim}), we empirically verify the results of Theorems \ref{thm:exp} 
and \ref{thm:pac} and illustrate that Algorithm \ref{alg:greedy} performs favorably compared to the competing 
greedy scheme both on average and for each individual sampling task.

Finally, in Theorem \ref{thm:curv} we extend the results of \cite{chamon2017greedy} derived for stationary 
graph signals and show that the maximum element-wise curvature of $f(\S) = \mathrm{Tr}(\P-\bar{\Sb}_\S)$ 
is bounded even for non-stationary graph signals and in the scenario where the statistics of the noise 
varies across nodes of the graph.
\begin{theorem}\label{thm:curv}
\textit{Let $\mathcal{C}_{f}$ be the maximum element-wise curvature of $f(\S) = \mathrm{Tr}(\P-\bar{\Sb}_\S)$. 
Then it holds that}
\begin{equation}\label{eq:curvbound}
\mathcal{C}_{\max} \leq\max_{j \in \mathcal{N}} \frac{\lambda_{\max}^2(\P)}{\lambda_{\min}^2(\P)}\left(1+\frac{\lambda_{\max}(\P)}{\sigma_j^2}\right)^3.
\end{equation}
\end{theorem}
\begin{proof}
	See Appendix \ref{pf:curv}.
\end{proof}
It was shown in \cite{chamon2017greedy} that if $\x$ is stationary and $\P = \sigma_\x^2\I_N$ for some 
$\sigma_\x^2>0$ and $\sigma^2_j = \sigma^2$ for all $j\in\mathcal{N}$, then the curvature of the MSE 
objective is bounded. However, Theorem \ref{thm:curv} holds even in the scenarios where the signal is 
non-stationary and the noise is not white.
\section{Numerical Examples}\label{sec:sim}
To assess the proposed support recovery and sampling algorithms, we study their performance in recovery of signals supported on synthetic and real-world graphs. In the first two subsections, we benchmark the performance of Algorithm 1, while in the rest of the subsections, we focus on evaluating the efficacy of the proposed randomized greedy algorithm.
\subsection{Synthetic Erd\H{o}s-R\'enyi random graphs I} \label{sec:erdo1} 
 \begin{figure*}[t]
	\begin{minipage}[t]{0.49\textwidth}
		\centering
		\includegraphics[width=1\linewidth]{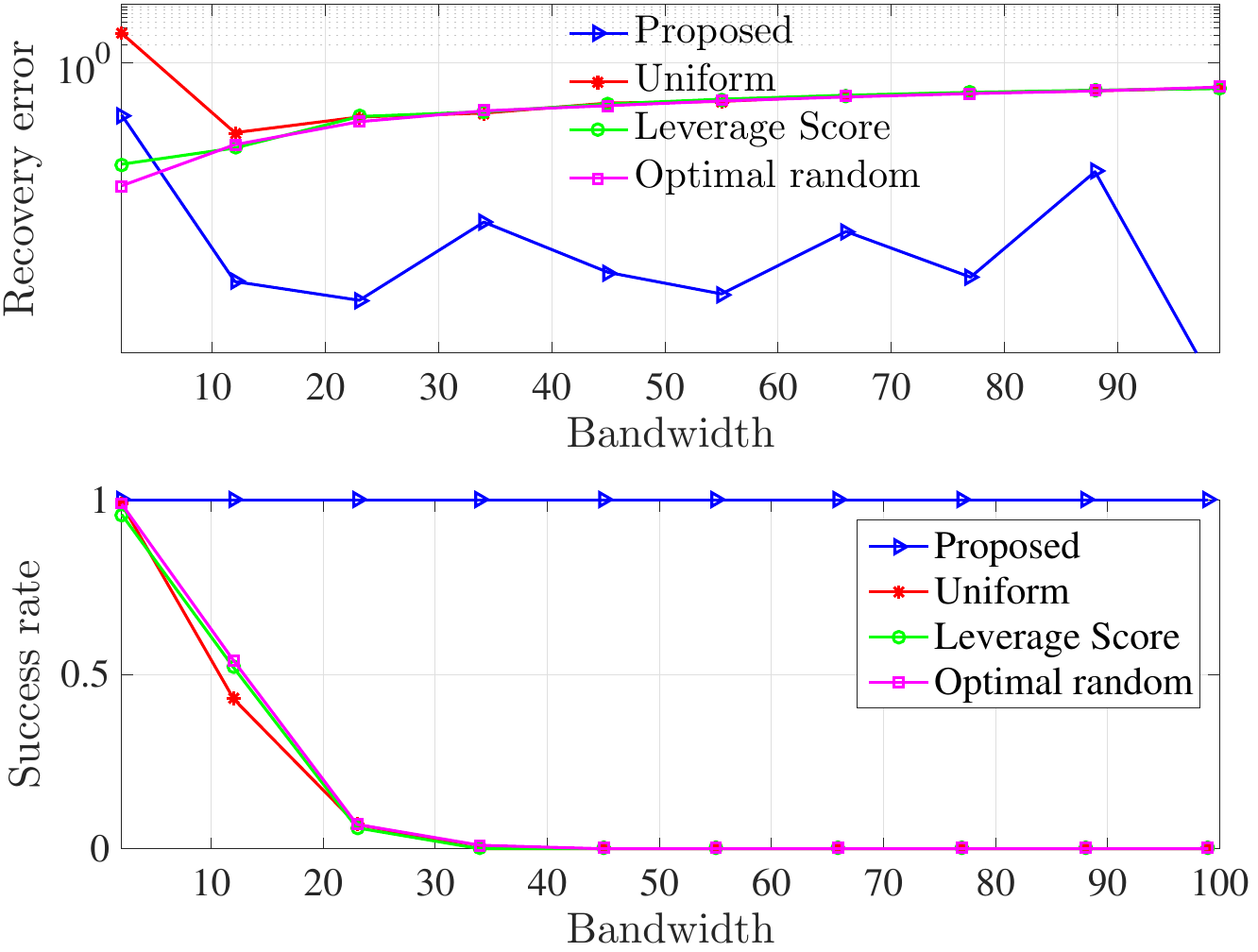}
		\centerline{\footnotesize(a)}\medskip
	\end{minipage}
	\hfill
	\begin{minipage}[t]{0.49\textwidth} 
		\centering
		\includegraphics[width=1\linewidth]{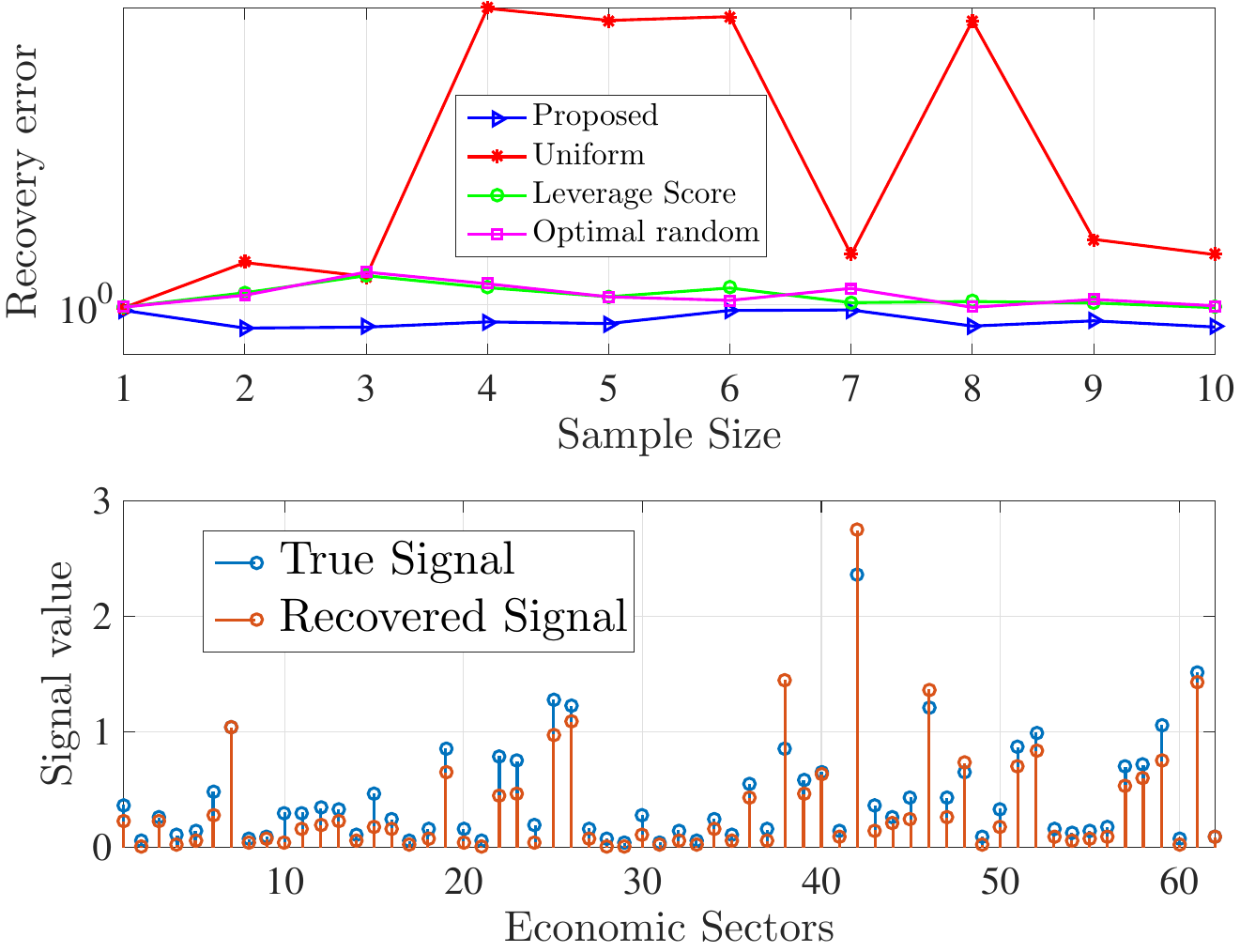}
		\centerline{\footnotesize (b)}\medskip
	\end{minipage}
	\caption{(a) Recovery error (top) and success rate 
		(bottom) of Algorithm 1 and various random selection sampling schemes versus bandwidth ($k$) 
		for undirected Erd\H{o}s-R\'enyi random graphs. 
		(b) Top: Recovery error comparison of different 
		selection sampling schemes as a function of the sample size for the economy network. Bottom: 
		Recovered and true graph signals for various economic sectors using Algorithm 1.}
	\label{fres1}
\end{figure*}
We first consider the task of sampling and reconstruction of noise-corrupted bandlimited graph signals with known support. Specifically, we consider undirected Erd\H{o}s-R\'enyi random graphs $\mathcal{G}$ of size $N = 100$ and edge probability $0.2$. We generate $\mathbf{x} = \mathbf{U}\bar{\mathbf{x}}_{K}$ by forming $\mathbf{U}$ using the first $k$ eigenvectors of the graph adjacency matrix, where $k$ is varied linearly from $2$ to $99$. The non-zero frequency components $\bar{\mathbf{x}}_{K}$ are drawn independently from a zero-mean Gaussian distribution with standard deviation $100$. The signal is corrupted by a Gaussian noise term with $\mathbf{Q} = 0.02^2\mathbf{I}_N$. We compare the recovery performance of the proposed scheme in Algorithm~\ref{alg:1} with state-of-the-art uniform, leverage score, and optimal random sampling schemes \cite{chen2015discrete,chen2016signal,puy2016random}. We define the recovery error as the ratio of the error energy to the true signal's energy.  Furthermore, the {\it success rate} \cite{chen2015discrete} is defined as the fraction of instances where $\C\U$ is invertible [cf. \eqref{eq:rec1}]. The results, averaged over 100 independent instances, are shown in Fig \ref{fres1}(a). As we can see from Fig \ref{fres1}(a) (top), 
the proposed scheme consistently achieves lower recovery error than competing schemes. Moreover, as shown in Fig \ref{fres1}(a) (bottom), when the bandwidth increases the success rate of random sampling schemes decreases while the success rate of the proposed scheme is always one, as formally established in Theorem \ref{thm:iss}.

Next, we compare the proposed sampling algorithm with Algorithm 1 of \cite{chen2015discrete} (see Fig~\ref{fig:new}) for  undirected
	Erd\H{o}s-R\'enyi  random graphs where we consider smaller bandwidth here to accommodate the computational cost of Algorithm 1 of \cite{chen2015discrete}. A that disadvantage of Algorithm 1 of \cite{chen2015discrete} compared to our method is that the iterative method of \cite{chen2015discrete} needs to perform singular value decomposition in each iteration to identify the sampling operator (see step 2 of Algorithm 1 in \cite{chen2015discrete}). Additionally, similar to our scheme which requires a residual node for initialization, \cite{chen2015discrete}  also needs an initial node. However, the selection of such an initial node is unclear in Algorithm 1 of \cite{chen2015discrete}. One major benefit of our method is that, as we show in Theorem 1, the proposed scheme achieves perfect recovery while  Algorithm 1 of \cite{chen2015discrete} does not have this important property. In terms of the empirical comparison, as Fig~\ref{fig:new} shows, the proposed iterative algorithm achieves a lower reconstruction error while consistently achieving success rate of one.
	
\begin{figure*}[t]
	\begin{minipage}[t]{0.49\textwidth}
		\centering
		\includegraphics[width=1\linewidth]{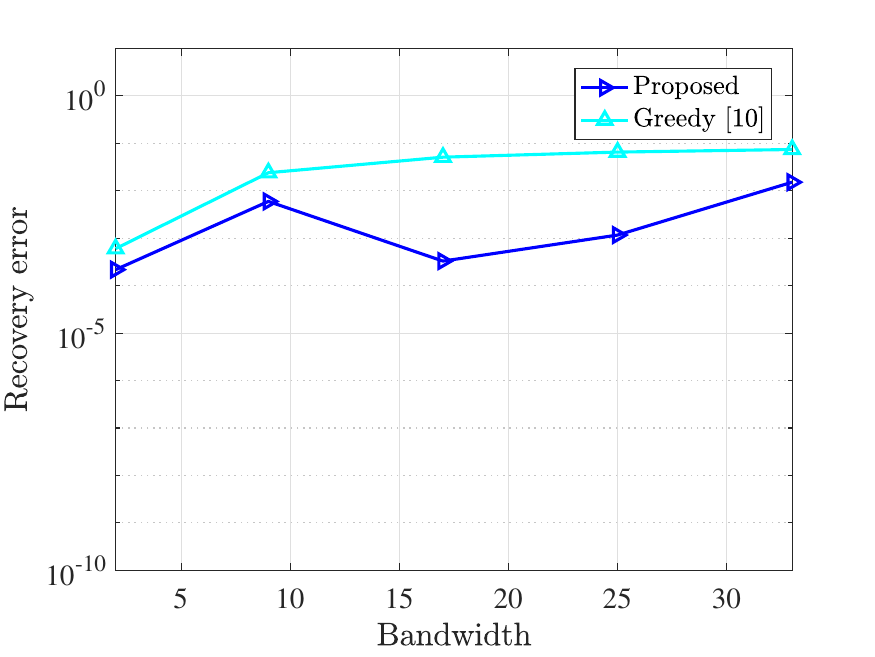}
		\centerline{\footnotesize(a)}\medskip
	\end{minipage}
	\hfill
	\begin{minipage}[t]{0.49\textwidth} 
		\centering
		\includegraphics[width=1\linewidth]{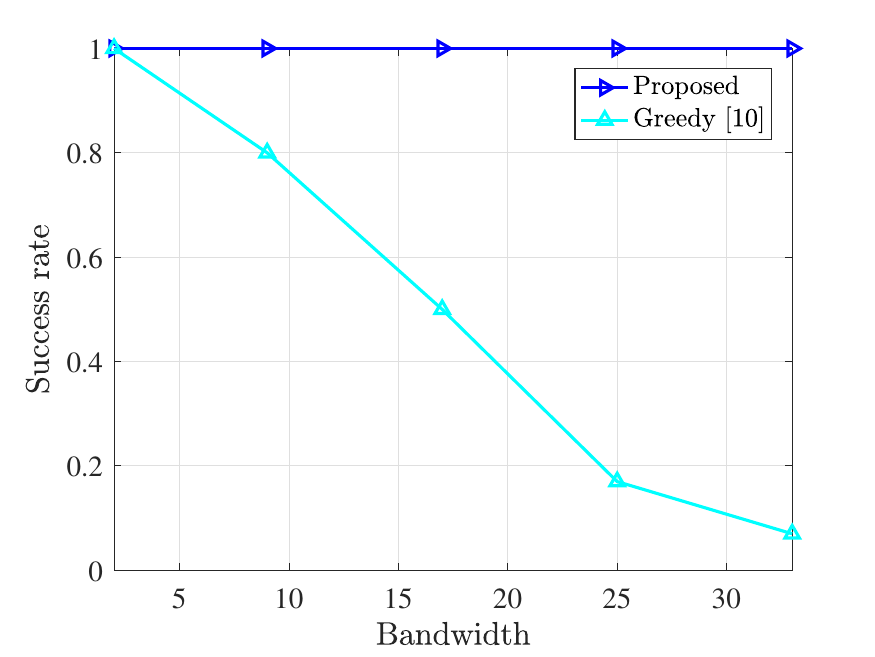}
		\centerline{\footnotesize (b)}\medskip
	\end{minipage}
	\caption{(a) Recovery error and (b) success rate 
		(bottom) of Algorithm 1 and the Greedy method of  \cite{chen2015discrete} versus bandwidth ($k$) 
		for undirected Erd\H{o}s-R\'enyi random graphs.}
	\label{fig:new}
\end{figure*}	
\subsection{Real graph: interpolation of industrial sectors' production} 
Next, we analyze data from the Bureau of Economic Analysis of the U.S. Department of 
Commerce which publicizes an annual table of input and outputs organized by economic sectors 
\footnote{Dataset from https://www.bea.gov.}. Specifically, we represent by nodes $62$ industrial 
sectors as defined by the North American Industry Classification System, and construct weighted 
edges and the graph signal similar to \cite{marques2016sampling}. The (undirected) edge weight 
between two nodes represents the average total production of the sectors, the first sector being 
used as the input to the other sector, expressed in trillions of dollars per year. This edge weight is 
averaged over the years $2008$, $2009$, and $2010$. Also, two artificial nodes are connected to 
all $62$ nodes as the added value generated and the level of production destined to the market of 
final users. Thus, the final graph has $N = 64$ nodes. The weights lower than $0.01$ are 
thresholded to zero and the eigenvalue decomposition of the corresponding adjacency matrix 
$\mathbf{A} = \mathbf{V} \mathbf{\Lambda} \mathbf{V}^{\top}$ is performed. A graph signal 
$\mathbf{x} \in \mathbb{R}^{64}$ can be regarded as a unidimensional total production -- in trillion 
of dollars -- of each sector during the year 2011. Signal $\mathbf{x}$ is shown to be approximately 
(low-pass) bandlimited in \cite[Fig. \ref{fres2}(a)(top)]{marques2016sampling} with a bandwidth of $4$.

We interpolate sectors' production by observing a few nodes using Algorithm~\ref{alg:1} and 
assuming that the signal is low-pass (i.e., with smooth variations over the built network). Then, 
we vary the sample size and compare the recovery performance of the proposed scheme with 
state-of-the-art uniform, leverage score, and optimal random sampling schemes 
\cite{chen2015discrete,chen2016signal,puy2016random} averaged over $1000$ Monte-Carlo 
simulations as shown in Fig. \ref{fres1}(b) (top). As the figure indicates, the proposed algorithm outperforms 
uniform, leverage score, and optimal random sampling schemes 
\cite{chen2015discrete,chen2016signal,puy2016random}. However, Algorithm~\ref{alg:1} does
not achieve perfect recovery in this noiseless scenario because the signal is not truly bandlimited.
Moreover, Fig. \ref{fres1}(b) (bottom) shows a realization of the graph signal $\mathbf{x}$ superimposed 
with the reconstructed signal obtained using Algorithm~\ref{alg:1} with $k=2$ for all nodes excluding 
two artificial ones. The recovery error of the reconstructed signal is approximately $1.32\%$; as 
Fig. \ref{fres1}(b) (bottom) illustrates, $\hat{\mathbf{x}}$ closely approximates $\mathbf{x}$.
\begin{figure*}[t]
	\centering
	\minipage[t]{1\linewidth}
	\begin{subfigure}[t]{.32\linewidth}
		\includegraphics[width=\textwidth]{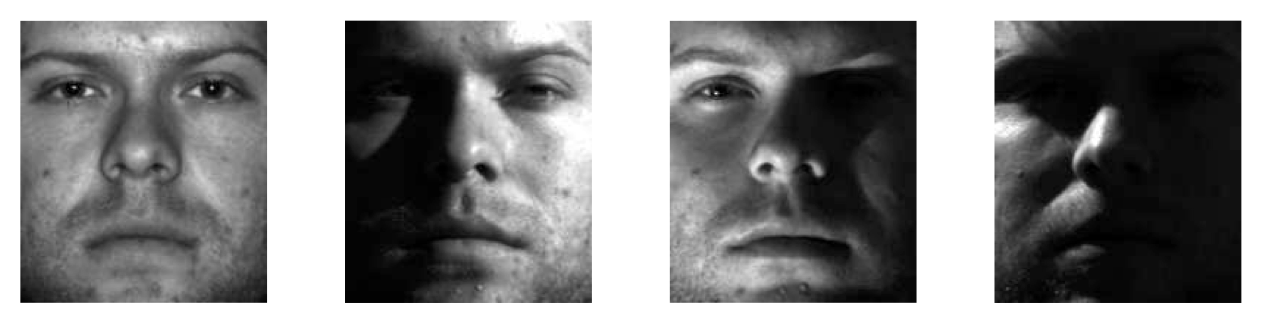}
	\end{subfigure}
	\begin{subfigure}[t]{.32\linewidth}
		\includegraphics[width=\linewidth]{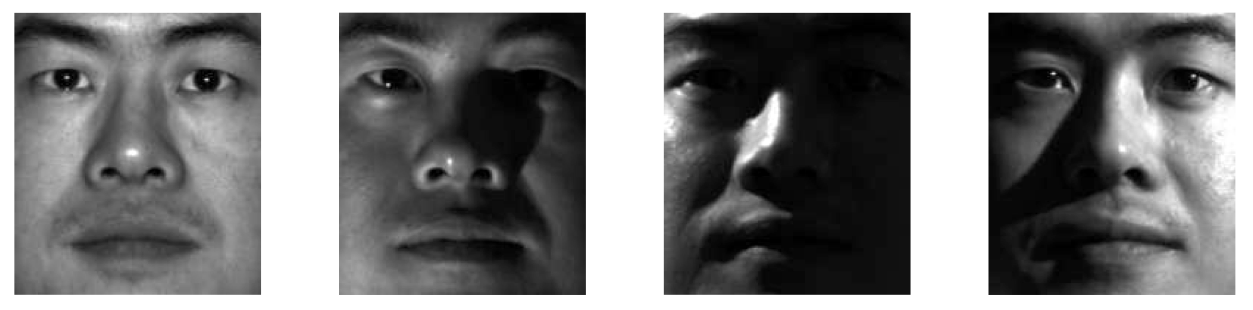}
	\end{subfigure}
	\begin{subfigure}[t]{.32\linewidth}
		\includegraphics[width=\linewidth]{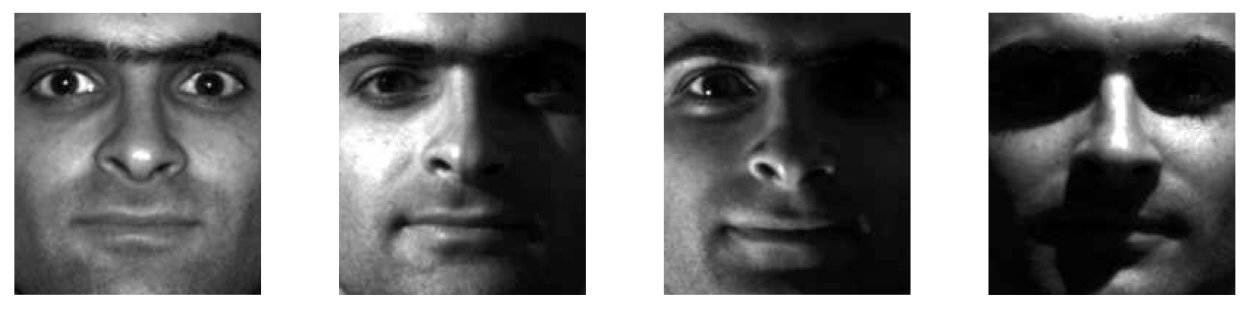}
	\end{subfigure}
	\caption{Face clustering: given images of multiple subjects, the goal is to find images that belong to the same subject 
	(examples from the EYaleB dataset \cite{georghiades2001few}).}
	\label{motionseg}
	\endminipage 
\end{figure*} 
\subsection{Synthetic graph: Localization of UAVs}
We now tackle a UAV localization problem in which the goal is to estimate absolute positions of 
robots from on-board sensor measurements. Specifically, consider a network of $N$ UAVs moving 
in a 2D plane and assume that each UAV is equipped with two systems: a laser scanner that 
measures the relative position of other UAVs within a sensing radius, and a GPS system that finds 
the absolute position of the UAV. \footnote{Notice that the graph structure in this application is essentially the time-varying communication network between the UAVs. {\color{black}In our simulation studies, we consider the localization task for only a single time-step. Nonetheless, the proposed sampling scheme can be employed in every time step where identification UAVs with GPS turned-on is required.}} While the laser system can find relative positions of the nearby 
UAVs with minimal power consumption, the GPS system requires intensive power to receive the 
location of the UAV from the control unit located potentially far from the network of UAVs. We 
consider the scenario where such inherent energy constraints prevent some UAVs 
to collect GPS data, i.e., only a subset of the UAVs can use the GPS. The objective is to compute 
the most representative subset of the UAVs so to minimize the MSE of the estimated global 
positions of all UAVs. To this end, we employ the proposed randomized-greedy scheme in 
Algorithm 2 with various values of $\epsilon$ to find a sampling set (a subset of UAVs) and 
compare its recovery error to that of the greedy sampling scheme \cite{chamon2017greedy}. 
Note that two graph 
signals, namely the x and y coordinates of UAVs, are supported on the network. Further, since 
UAVs that are close to each other have similar locations, both of these graph signals are smooth 
and hence bandlimited. {\color{black}It is worth to note that the collection of UAVs, typically referred to as UAV swarm, has a {\it swarm leader} that is task with handling costlier computations and is capable of communicating with the control unit that guides the swarm in moving in the environment.
}

We run Monte Carlo simulations with 1000 instances where we consider 1000 UAVs distributed 
uniformly on a $10\times 10$ grid; the range of the laser system is set to 0.3 and the power of noise 
affecting laser measurements is set to $10^{-2}$. The recovery error and running time results as a 
function of signals' bandwidth -- which is also the size of the sampling set -- are shown in Fig. \ref{fres2}(a) 
and Fig. \ref{fres2}(b), respectively. As we see in Fig. \ref{fres2}(a), performance of the proposed 
scheme and the greedy 
algorithm are fairly similar; as bandwidth increases, the recovery error decreases. Furthermore, 
as $\epsilon$ gets smaller, the gap between the performance of the proposed scheme and the 
greedy algorithm reduces until becoming negligible. The running time comparison illustrated in 
Fig. \ref{fres2}(b) reveals that for the largest sampling set considered (i.e. $k = 50$), the proposed scheme 
is more than 2x faster than the greedy method. Additionally, the complexity of the proposed scheme 
is linear in $k$, while that of the greedy method is quadratic, as predicted by our theoretical results; see also \cite{hashemi2018sampling} for additional MSE performance and runtime comparisons with the greedy sampling algorithm in \cite{chamon2017greedy}.
\subsection{Real graph: Semi-supervised face clustering} 
Clustering faces of individuals is an important task in computer vision 
\cite{georghiades2001few,you2015sparse,hashemi2017accelerated,hashemi2018evolutionary}. In real-world settings, 
labeling all images is practically infeasible. However, acquiring labels even for a small subset 
of data that can represent all images may drastically improve the clustering accuracy. The proposed randomized-greedy selection sampling framework can be employed in this setting to acquire labels for a small number of images to achieve improved clustering accuracy.  To this end, we test the randomized-greedy 
algorithm on EYaleB dataset \cite{georghiades2001few} (see Fig. \ref{motionseg}) 
which contains frontal face images of 38 individuals under 64 different illumination conditions. 
Similar to the prior works (see, e.g., \cite{you2015sparse,hashemi2017accelerated,hashemi2018evolutionary}), in our
studies the images are down-sampled to $48 \times 42$ from the original size of 
$192 \times 168$. In each of $100$ independent instances of the Monte Carlo simulation we 
randomly pick $8$ subjects and all of their images as the data points to be clustered; this 
results in a clustering problem with $N = 512$ data points. To construct the underlying graph signal
and capture similarity of the data points, we employ the sparse subspace clustering (SSC) 
scheme recently proposed in \cite{you2015sparse} to find the adjacency matrix $\A$ and the Laplacian matrix $\mathbf{L}$. The 
graph signal support on the constructed similarity graph is discrete valued, i.e., the value of each node is an integer 
in $\{1,\dots,8\}$. Note that the graph signal supported  on the constructed similarity  graph is smooth and bandlimited 
as similar images are unlikely to correspond to different individuals. The performance 
comparison of Algorithm 2 with various values for $\epsilon$, greedy sampling method, 
random sampling schemes, and the unsupervised clustering method are illustrated in Fig. \ref{fres2}(c) 
as a function of the sampling ratio ($k\slash N$). For the sake of clarity of presentation, we only 
show the result of the best method among uniform, leverage score, and optimal random sampling 
approaches \cite{chen2015discrete,puy2016random}. As we see in Fig. \ref{fres2}(c), the greedy and 
randomized-greedy schemes deliver the best clustering performance; as we increase size of 
the sampling set, the accuracy of semi-supervised schemes improves and the gap between the 
performance of random sampling methods and the proposed scheme decreases. Furthermore, 
our simulation studies reveal that acquiring labels of only 8 data points using the proposed 
scheme results in more than $12\%$ improvements in clustering accuracy as compared to the 
unsupervised method.
\begin{figure*}[t]
	\begin{minipage}[t]{0.32\textwidth}
		\centering
		\includegraphics[width=1\linewidth]{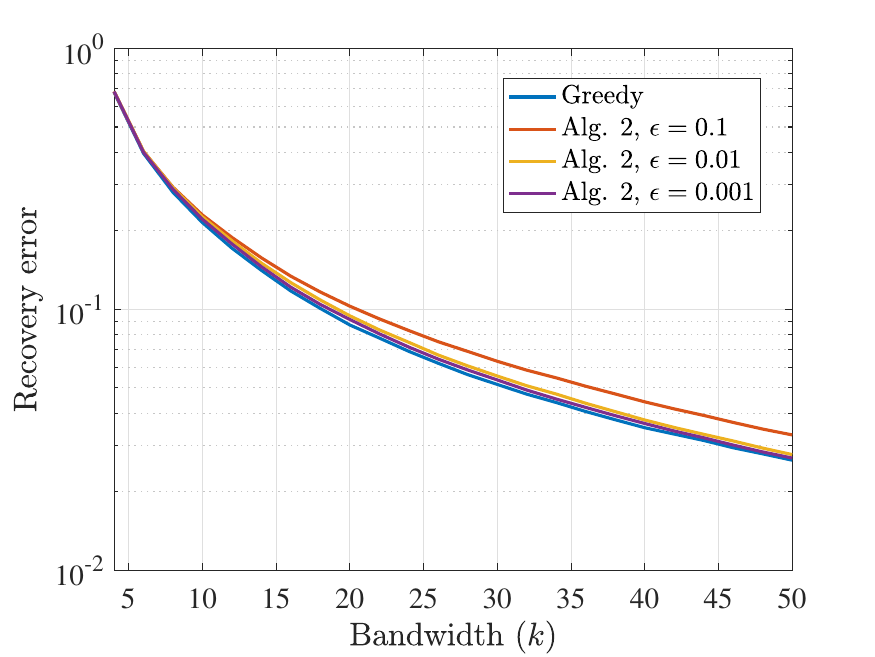}
		\centerline{\footnotesize(a)}\medskip
	\end{minipage}
	\hfill
	\begin{minipage}[t]{0.32\textwidth} 
		\centering
		\includegraphics[width=1\linewidth]{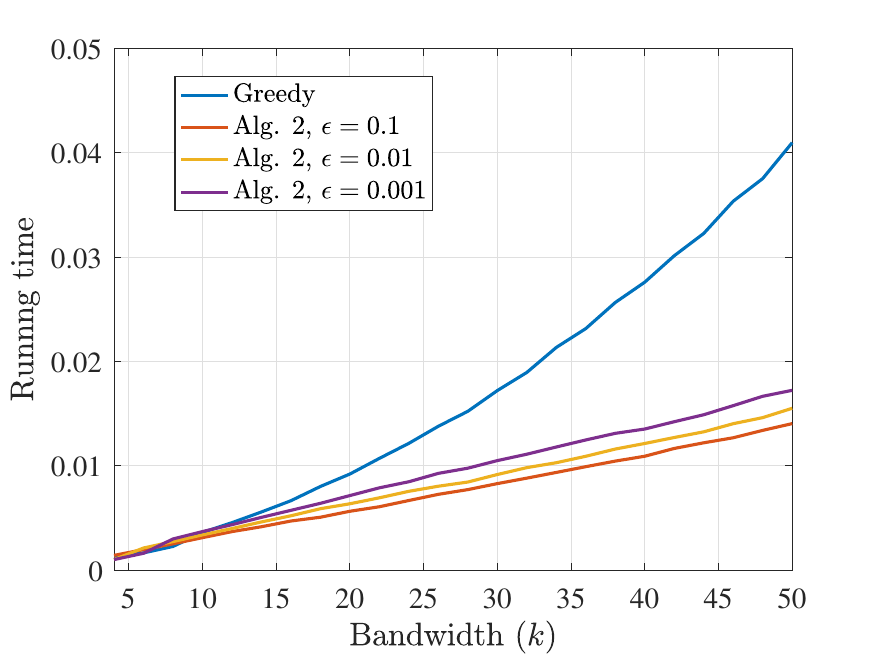}
		\centerline{\footnotesize(b)}\medskip
	\end{minipage}
	\hfill
	\begin{minipage}[t]{0.32\textwidth} 
		\centering
		\includegraphics[width=1\linewidth]{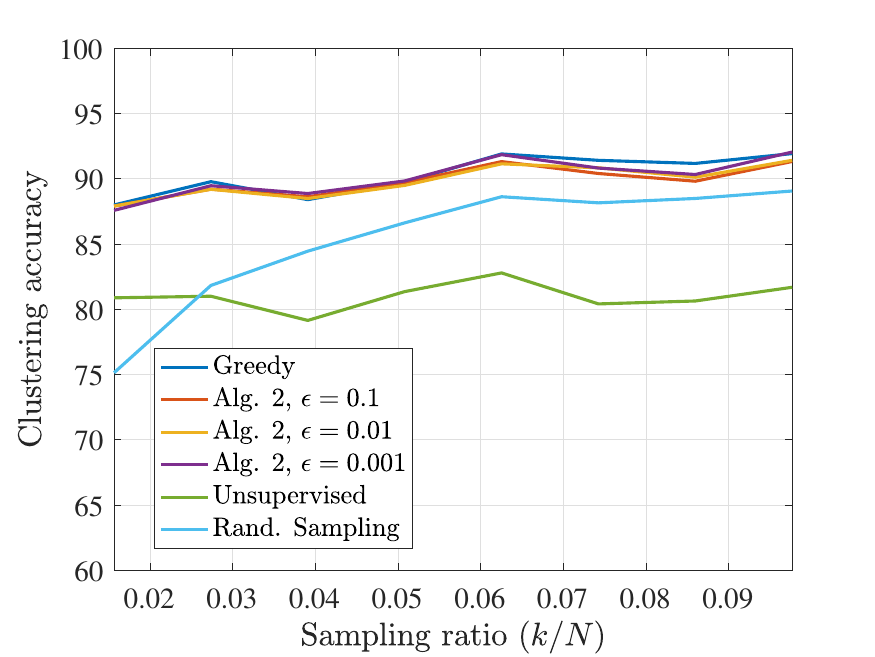}
		\centerline{\footnotesize (c)}\medskip
	\end{minipage}
	\caption{(a) Recovery error comparison of the greedy scheme \cite{chamon2017greedy} and Algorithm 2 as a function 
		of bandwidth for the UAV localization problem. (b) Running time comparison of the greedy 
		scheme \cite{chamon2017greedy} and Algorithm 2 as a function of bandwidth for the UAV localization problem. 
		(c) Clustering accuracy of greedy \cite{chamon2017greedy}, Algorithm 2, random sampling, and unsupervised methods as a function of the sampling ratio for the face 
		clustering application.}
	\label{fres2}
\end{figure*}
\subsection{Synthetic Erd\H{o}s-R\'enyi random graphs II}  
Since Algorithm 2 is a randomized scheme, in this section we study the performance of Algorithm 2 for each individual sampling tasks (i.e. each Monte-Carlo realizations). To this end, we again consider the Erd\H{o}s-R\'enyi random graphs, similar to those in Section \ref{sec:erdo1}. Here, we study the setting where $N=10$ and $k=4$. 
Bandlimited graph signals are generated as before except that this time we take $\mathbf{U}$ as the first $4$ eigenvectors of the adjacency matrix. Figs.~\ref{fig:min} (a) depicts superimposed MSE histograms of Algorithm 2 and the greedy sampling scheme \cite{chamon2017greedy} for 100 realizations per method and fixed $|S|=4$. As the figure illustrates, the proposed randomized greedy schemes performs well and is comparable with the greedy approach, not just on average but also for majority of individual sampling tasks.
\subsection{Real graph: Minnesota road network}
Next, we consider the Minnesota road network\footnote{https://sparse.tamu.edu/Gleich/minnesota} with $N=2642$ nodes in order to showcase scalability of the proposed graph sampling method. To that end, Bandlimited graph signals are generated by taking the first $k=600$ eigenvectors of the graph Laplacian matrix, where the non-zero frequency components are drawn from a zero-mean, multivariate Gaussian distribution with randomly chosen PSD covariance matrix $\P$. The signals are corrupted with additive white Gaussian noise with $\sigma^{2}=10^{-2} \mathbf{I}_{N}$. 
As expected, Figs.~\ref{fig:min} (b) and (c) depict trends of decreasing MSE and increasing running time versus $|S|$, respectively. The results are averaged over $1000$ Monte-Carlo simulations run. Remarkably, the proposed randomized greedy procedure achieves an order-of-magnitude speedup over the state-of-the-art algorithm in \cite{chamon2017greedy} while showing only a marginal degradation in the MSE performance. {\color{black}Note that the time of performing eigenvalue decomposition to find the graph shift operator $\U$ in MATLAB was less than 2 seconds on a typical laptop. Figs.~\ref{fig:min} (d) depicts the runtime comparison of the proposed scheme versus the benchmark by accounting for the time of computing the eigenvalue decomposition.}
\begin{figure*}[t]
	\begin{minipage}[t]{0.49\textwidth}
		\centering
		\includegraphics[width=1\linewidth]{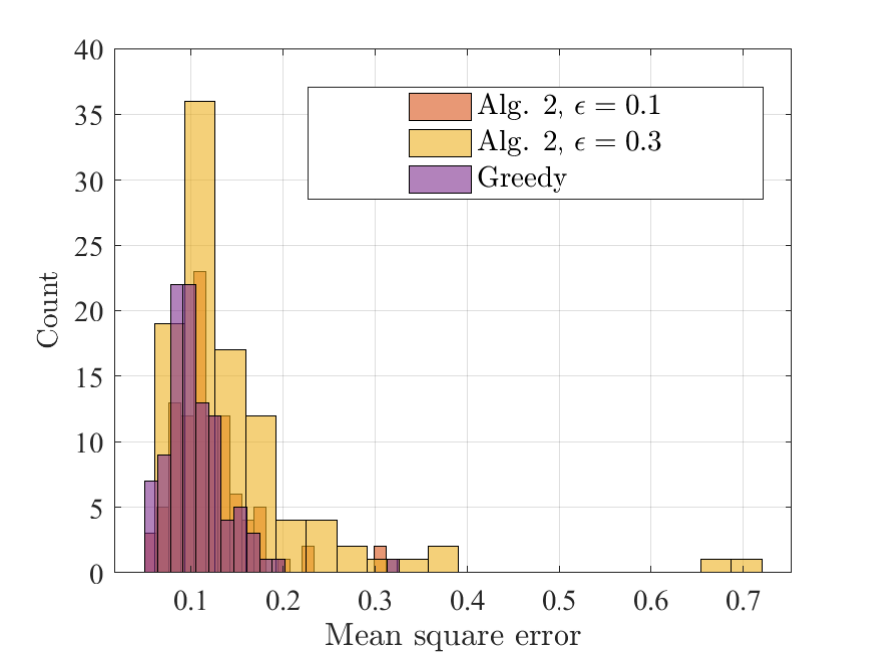}
		\centerline{\footnotesize(a)}\medskip
	\end{minipage}
	\hfill
	\begin{minipage}[t]{0.49\textwidth} 
		\centering
		\includegraphics[width=1\linewidth]{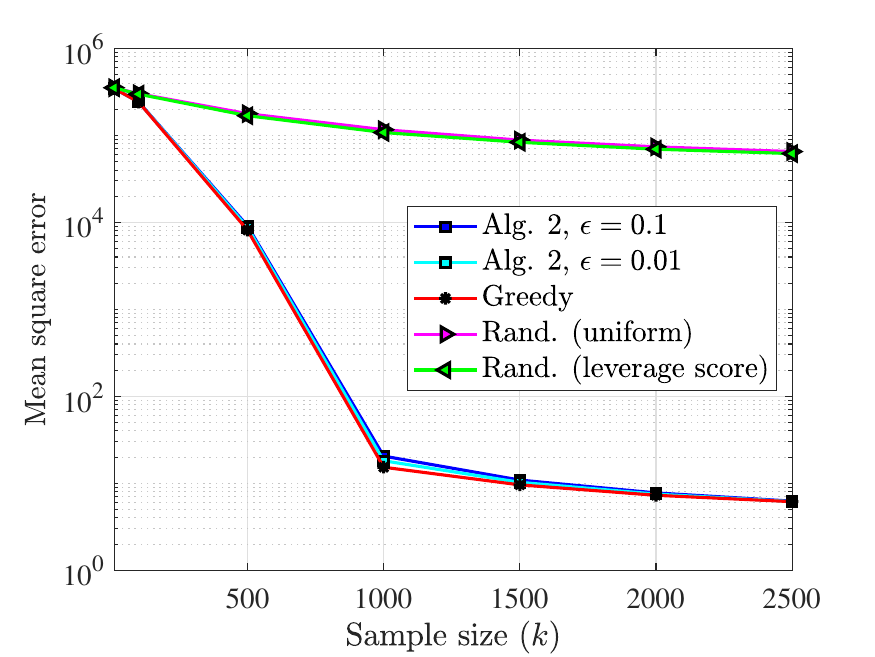}
		\centerline{\footnotesize(b)}\medskip
	\end{minipage}
	\hfill
	\begin{minipage}[t]{0.49\textwidth} 
		\centering
		\includegraphics[width=1\linewidth]{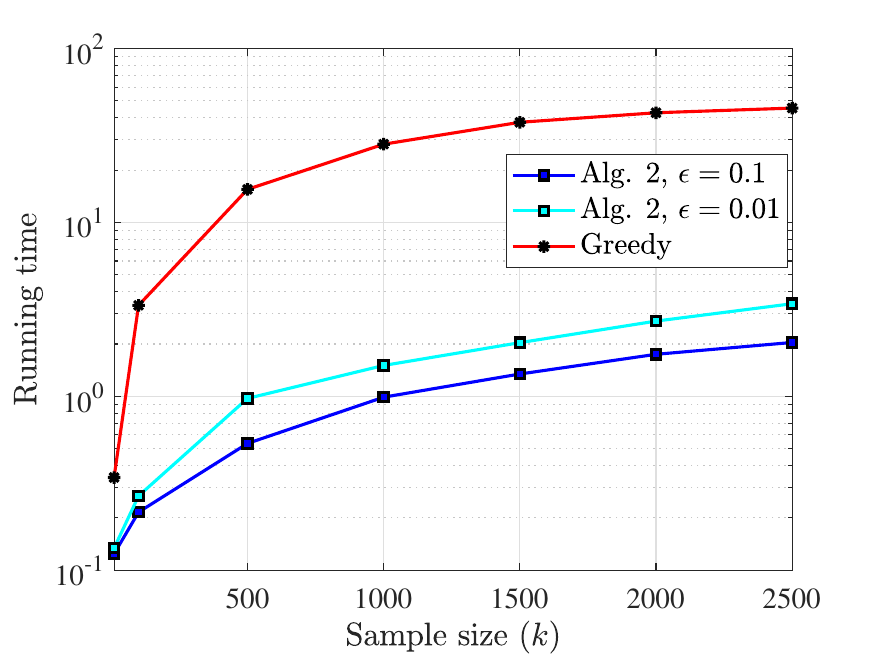}
		\centerline{\footnotesize (c)}
	\end{minipage}
	\begin{minipage}[t]{0.49\textwidth} 
		\centering
		\includegraphics[width=1\linewidth]{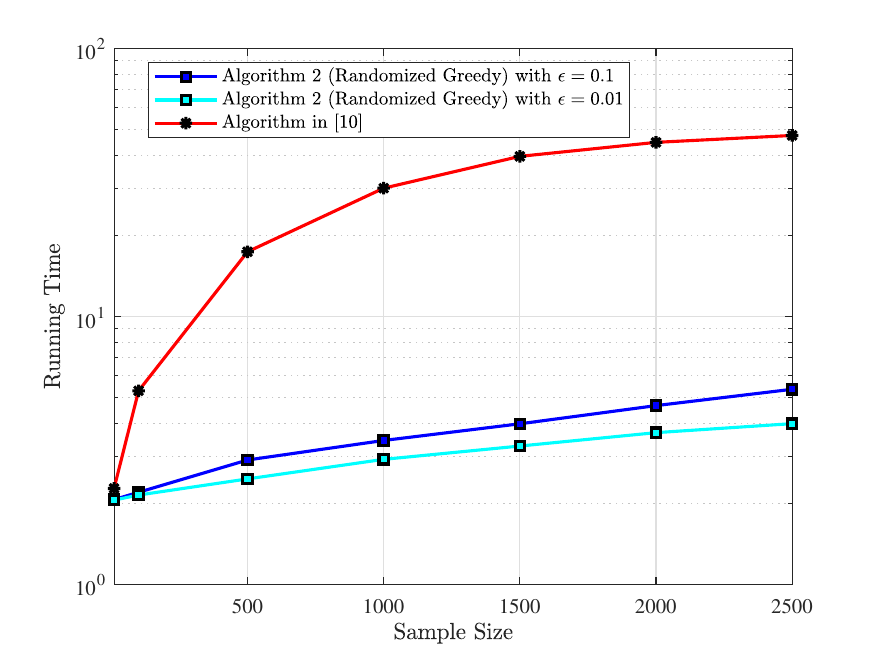}
		\centerline{\footnotesize (d)}\medskip
	\end{minipage}
	\caption{(a) Histogram of MSE values for $100$ realizations and fixed sampling set size in simulated Erd\H{o}s-R\'enyi graphs. (b) MSE comparison of greedy \cite{chamon2017greedy}, Algorithm 2, and random sampling schemes on Minnesota road network. (c) Running time comparison of the greedy scheme \cite{chamon2017greedy} and Algorithm 2 on Minnesota road network, excluding the time of eigenvalue decomposition. (d) Running time comparison of the greedy scheme \cite{chamon2017greedy} and Algorithm 2 on Minnesota road network, including the time of eigenvalue decomposition.}
	\label{fig:min}
\end{figure*}
{\color{black}
\subsection{Synthetic graph: Large-Scale preferential attachment random graph}
Finally, we consider a large-scale preferential attachment random graph \cite{barabasi1999emergence} with $N= 10,000$ nodes to show the superiority of the proposed Algorithm 2 over existing methods. In particular, similar to the previous random graph simulations, we generate random band-limited Gaussian graph signals using the first 500 eigenvectors of the preferential attachment graph adjacency matrix (see Fig. \ref{fig:pa} (c) for the structure of the sparse adjacency matrix).
 The results are illustrated in Fig. \ref{fig:pa} where as we see, Algorithm 2  achieves the same performance as that of the greedy scheme \cite{chamon2017greedy} while incurring orders of magnitude lower running time.

\begin{figure*}[t]
	\begin{minipage}[t]{0.32\textwidth}
		\centering
		\includegraphics[width=1\linewidth]{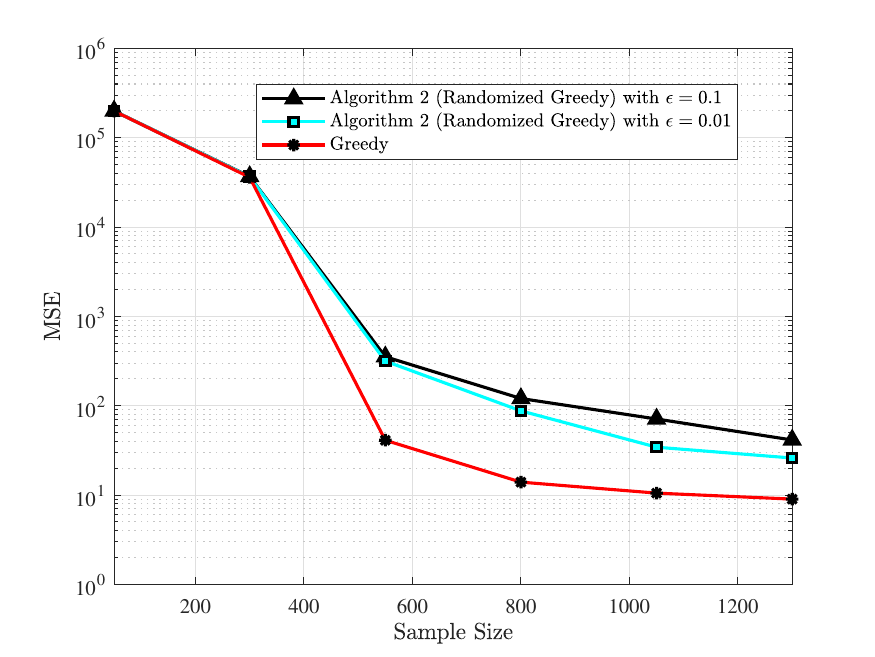}
		\centerline{\footnotesize (a) Reconstruction error}\medskip
	\end{minipage}
	\hfill
	\begin{minipage}[t]{0.32\textwidth} 
		\centering
		\includegraphics[width=1\linewidth]{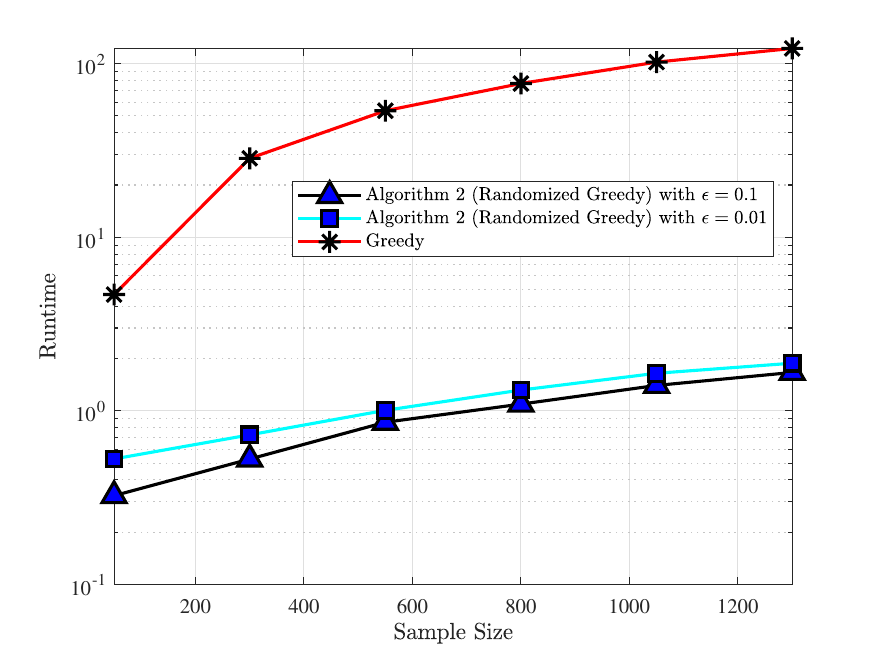}
		\centerline{\footnotesize (b) running time}\medskip
	\end{minipage}
	\hfill
	\begin{minipage}[t]{0.32\textwidth} 
		\centering
		\includegraphics[width=1\linewidth]{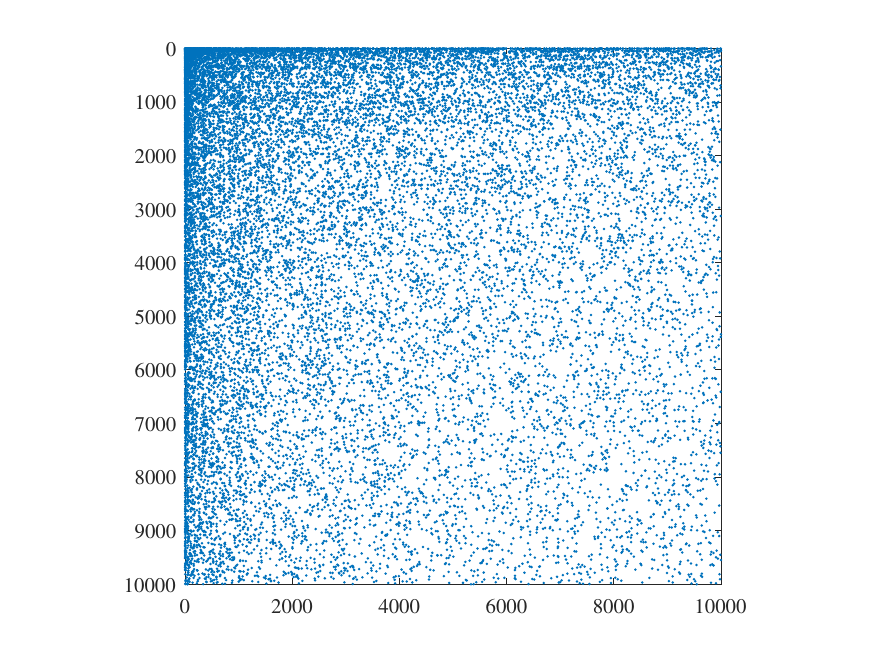}
		\centerline{\footnotesize (c) Adjacency matrix}\medskip
	\end{minipage}
	\caption{Performance comparison of greedy scheme \cite{chamon2017greedy} and Algorithm 2 on a large-scale preferential attachment random graph with $N= 10,000$ nodes.}
	\label{fig:pa}
\end{figure*}
}

\section{Conclusion} \label{sec:concl}
We considered the task of sampling and reconstruction of spectrally sparse graph signals. where the goal 
is to interpolate a (non-stationary) graph signal from a small subset of the nodes with 
the lowest reconstruction error. First, we studied the non-Bayesian scenario and proposed an efficient 
iterative sampling approach that exploits the low-cost selection criterion of the orthogonal matching pursuit 
algorithm to recursively select a subset of nodes of the graph. We then theoretically showed that in the 
noiseless case the original $k$-spectrally sparse signal is perfectly recovered from the set of selected nodes with 
cardinality $k$. In the case of noisy measurements, we established a worst-case performance bound on 
the reconstruction error of the proposed algorithm. In the Bayesian scenario where the graph signal is a 
non-stationary random process, we formulated the sampling task as the problem of maximizing a 
monotone weak submodular function that is directly related to the mean square error (MSE) of the linear 
estimator of the original signal. We proposed a randomized-greedy algorithm to find a sub-optimal subset 
of sampling nodes. By analyzing the performance of the randomized-greedy algorithm, we showed that 
the resulting MSE is a constant factor away from the MSE of the optimal sampling set. Unlike prior work, 
our guarantees do not require stationarity of the graph signal and the study is not restricted to the case 
of additive white noise. Instead, the noise coefficients are assumed to be independent but the power of 
noise varies across individual nodes of the graph.
Extensive simulations on on synthetic and real-world graphs with applications in economics, localization, and clustering showed that the proposed iterative and randomized-greedy selection sampling algorithms outperform the competing alternatives in terms of accuracy~and~runtime.
\section{Acknowledgment} We would like to thank the authors of \cite{marques2016sampling} for providing the data used for the economy network analysis. Work in this paper was supported in part by the NSF award CCF-1750428 and ECCS-1809327.
\begin{appendices}
\section{Proof of Theorem \ref{thm:iss}}\label{pf:iss}
To prove the theorem, it suffices to show that Algorithm \ref{alg:1} selects a subset of rows of $\U$ which are linearly independent. Consider the $i\ts{th}$ iteration where $\u_{s_i}$ is identified and assume that until this iteration $\S$ contains indices of a collection of linearly independent vectors $\{\u_{s_1},\dots,\u_{s_{(i-1)}}\}$. If $|\r_{i-1}^\top \u_{s_i}|\neq 0$, since  $\r_{i-1}$ is orthogonal to the span of $\{\u_{s_1},\dots,\u_{s_{(i-1)}}\}$,  $\u_{s_i}$ is not in the span of these vectors. Hence,  $\{\u_{s_1},\dots,\u_{s_i}\}$ is also a collection of linearly independent vectors and by an inductive argument we conclude that rows of $\U_{\S,r}$ are linearly independent. Now assume  $|\r_{i-1}^\top \u_{s_i}| =  0$ for some $i\leq k$. Since $\U$ does not have all-zero rows, this condition implies $\r_{i-1} = \mathbf{0}$.\footnote{We note that if $|\r_{i-1}^T \u_{s_i}| = 0$ for $i\leq k$, $\r_{i-1} \neq \mathbf{0}$,  and  $\r_{i-1}$ and $\u_{s_i}$ are orthogonal, then  since ${s_i}$ is the optimizer of the selection criterion in step 6, $\r_{i-1}$ is orthogonal to all $\u_j$ with $j\in \mathcal{N}\backslash\{\ell\} \backslash\S$. Now, since by definition $\r_{i-1}$ is orthogonal to the subspace spanned by nodes indexed by $\S$, we conclude that $\r_{i-1} \in \R^k$ is orthogonal to the subspace spanned by all $\u_j$, i.e. $\R^k$. However, this can only hold for $\r_{i-1} = \mathbf{0}$.} Therefore, all the remaining rows of $\U$ which are not selected lie in the span of $\{\u_{s_1},\dots,\u_{s_{(i-1)}}\}$. Since by assumption $i\leq k$, this condition implies that the rank of $\U$ is at most $k-1$ which contradicts the fact that $\V$ is a basis and $\U$ has full column-rank. Therefore, $\r_{i-1} =\mathbf{0}$ holds only for $i>k$ and thus rows of $\U_{\S,r}$ are linearly independent. This completes the proof.  
\section{Proof of Proposition \ref{thm:nois}}\label{pf:thmnois}
According to Theorem \ref{thm:iss}, if $m = k$, $\U_{\S,r} = \C\U$ is invertible. Therefore, since $\Q_\S$ is positive definite and invertible it is easy to see that $\U_{\S,r}^\top\Q_\S^{-1}\U_{\S,r}$ is also invertible. Now consider the case $m\geq k$ where $\U_{\S,r} \in\R^{m\times k} $ is a tall full rank matrix. Let $\Q_\S^{-1} = \mathbf{L} \mathbf{L}^\top$ be the Cholesky decomposition of  $\Q_\S^{-1}$. Since $\Q_\S^{-1}$ is a positive definite matrix, $\mathbf{L}\in\R^{m\times m}$ is full rank and invertible. Therefore, $\mathbf{L}_u = \mathbf{L}^\top\U_{\S,r}$ is  also a full rank matrix. Thus, $\U_{\S,r}^\top\Q_\S^{-1}\U_{\S,r} = \mathbf{L}_u^\top\mathbf{L}_u \in\R^{k \times k}$ is full rank and invertible. 
Hence, for any $m \geq k$ given a $\C$ constructed  by Algorithm \ref{alg:1}, \eqref{eq:normal} simplifies to
\begin{equation}\label{eq:normal1}
\hat{\bar{\x}} = (\U_{\S,r}^\top\Q_\S^{-1}\U_{\S,r})^{-1}\U_{\S,r}^\top \Q_\S^{-1} \tilde{\x}.
\end{equation}
Since  $\x = \U\bar{\x}_\K$, the reconstructed signal $\hat{\x}$ can be obtained according to 
\begin{equation}\label{eq:normal2}
\begin{aligned}
\hat{\x} &= \U(\U_{\S,r}^\top\Q_\S^{-1}\U_{\S,r})^{-1}\U_{\S,r}^\top \Q_\S^{-1} \tilde{\x}\\
& = \U(\U_{\S,r}^\top\Q_\S^{-1}\U_{\S,r})^{-1}\U_{\S,r}^\top \Q_\S^{-1}( \U_{\S,r}\bar{\x}_\K +\n_{\S})\\
& =  \U\bar{\x}_\K+\U(\U_{\S,r}^\top\Q_\S^{-1}\U_{\S,r})^{-1}\U_{\S,r}^\top \Q_\S^{-1}\C\n\\
& = \x +\U(\U_{\S,r}^\top\Q_\S^{-1}\U_{\S,r})^{-1}\U_{\S,r}^\top \Q_\S^{-1}\n_\S.
\end{aligned}
\end{equation}
Therefore,
\begin{equation}\label{eq:normal3}
\begin{aligned}
\|\hat{\x}-\x\|_2 &\leq \|\U(\U_{\S,r}^\top\Q_\S^{-1}\U_{\S,r})^{-1}\U_{\S,r}^\top \Q_\S^{-1}\n_\S\|_2\\
&\stackrel{(a)}{\leq}\|\U(\U_{\S,r}^\top\Q_\S^{-1}\U_{\S,r})^{-1}\U_{\S,r}^\top \Q_\S^{-1}\n\|_2\\
&\stackrel{(b)}{\leq}\sigma_{\max}(\U(\U_{\S,r}^\top\Q_\S^{-1}\U_{\S,r})^{-1}\U_{\S,r}^\top \Q_\S^{-1})\epsilon_\n\\
&\stackrel{(c)}{\leq}\sigma_{\max}(\U)\sigma_{\max}((\U_{\S,r}^\top\Q_\S^{-1}\U_{\S,r})^{-1}\U_{\S,r}^\top \Q_\S^{-1})\epsilon_\n\\
&\stackrel{(d)}{\leq}\sigma_{\max}((\U_{\S,r}^\top\Q_\S^{-1}\U_{\S,r})^{-1}\U_{\S,r}^\top \Q_\S^{-1})\epsilon_\n\\
\end{aligned}
\end{equation}
where $(a)$ and $(b)$ follow by the assumption $\|\n_\S\|_2 \leq \|\n\|_2 \leq \epsilon_\n$, $(c)$ 
stems from submultiplicative property of $\ell_2$-norm, and $(d)$ is by the fact that 
$\sigma_{\max}(\U) = 1$ as it is a submatrix of an orthogonal matrix.
\section{Proof of Proposition \ref{thm:p}}\label{pf:mono}
We first verify that
\[f(\emptyset)=\mathrm{Tr}\left(\P-\bar{\Sb}_{\emptyset}\right)=\mathrm{Tr}\left(\P-\P\right)=0.\]
Next, to show monotonicity, we establish a recursive relation for the marginal gain of selecting a new node on graph. More specifically, for $j\in[n]\backslash \S$ it holds that
\begin{equation}\label{eq:update}
\begin{aligned}
f_j(\S)&= \mathrm{Tr}\left(\P-\bar{\Sb}_{\S\cup \{j\}}\right) -  \mathrm{Tr}\left(\P-\bar{\Sb}_\S\right)\\
&= \mathrm{Tr}\left(\bar{\Sb}_\S\right) -\mathrm{Tr}\left(\bar{\Sb}_{\S\cup \{j\}}\right)\\
&=  \mathrm{Tr}\left(\bar{\Sb}_\S\right)-\mathrm{Tr}\left(\left(\bar{\Sb}_\S^{-1}+\sigma_j^{-2}\u_j\u_j^\top\right)^{-1}\right)\\
&\stackrel{(a)}{=} 
\mathrm{Tr}\left(\frac{\bar{\Sb}_\S\u_j\u_j^\top\bar{\Sb}_\S}{\sigma_j^{2}+\u_j^\top\bar{\Sb}_\S\u_j}\right) \stackrel{(b)}{=} \frac{\u_j^\top\bar{\Sb}_\S^2\u_j}{\sigma_j^{2}+\u_j^\top\bar{\Sb}_\S\u_j}
\end{aligned}
\end{equation}
where $(a)$ easily follows by applying Sherman–Morrison formula \cite{bellman1997introduction} 
on matrix $(\bar{\Sb}_\S^{-1}+\sigma_j^{-2}\u_j\u_j^\top)^{-1}$, and $(b)$ is due to properties of 
the trace of a matrix. Finally, since $\bar{\Sb}_\S$ is the error covariance matrix, it is symmetric 
and positive definite. Hence, $f_j(\S) > 0$, which in turn implies monotonicity.
\section{Proof of Theorem \ref{thm:exp}}\label{pf:exp}
To prove the stated results, we first we state 
Lemma \ref{lem:curv} \cite{ma2} that upper-bounds the difference between the values of the objective 
corresponding to two sets having different cardinalities.

\begin{lemma}\label{lem:curv}
	\cite{ma2} Let $f$ denote a monotone set function with the maximum element-wise curvatures ${\cal C}_{max}$. 
	Let $\S$ and $\T$ be any two sampling sets such that $\S\subset \T \subseteq \mathcal{N}$ with 
	$|\T\backslash \S|=r$. Then, it holds that
	\begin{equation}
	f(\T)-f(\S)\leq  C(r)\sum_{j\in \T\backslash \S}f_j(\S),
	\end{equation}
	where $ C(r)=\frac{1}{r}(1+(r-1){\cal C}_{f})$. Moreover, $C(r)$ is decreasing (increasing) with respect to 
	$\mathcal{R}$ if ${\cal C}_{f}<1$ (${\cal C}_{f}>1$).
\end{lemma}
To prove the theorem, we first establish a bound on the expected value of the marginal gains 
of adding new nodes to the sampling set. Then, using the results of Lemma \ref{lem:curv}, we 
reduce the proof of approximation factor to that of the classical greedy algorithm introduced in \cite{nemhauser1978analysis}. More specifically, consider the $i\ts{th}$ iteration of Algorithm \ref{alg:greedy} and let $\S$ and $(i+1)_s$ denote the current sampling set and the index of node selected at the $(i+1)\ts{st}$ iteration of Algorithm \ref{alg:greedy}. A necessary condition to 
achieve the optimal MSE is that set $\mathcal{R}$ at each iteration must contain at least one 
node from the optimal sampling set $\O$. Let $\Phi =\mathcal{R} \cap (\O\backslash \S)$.  Since $\mathcal{R}$ is generated via sampling without replacement, it holds that
\begin{equation}
\begin{aligned}
\Pr\{\Phi = \emptyset\}& = \prod_{l = 0}^{s-1} \left(1-\frac{|\O\backslash \S|}{|\mathcal{N}\backslash \S|-l}\right)\\
&\stackrel{(a)}{\leq}\left(1-\frac{|\O\backslash \S|}{s}\sum_{l=0}^{s-1}\frac{1}{N-l}\right)^s\\
&\stackrel{(b)}{\leq} (1-\frac{|\O\backslash \S|}{s}(H_N-H_{N-s}))^s 
\end{aligned}
\end{equation}
where $(a)$ is by the inequality between arithmetic and geometric means and the fact that 
$|\mathcal{N}\backslash \S|\leq N$, and 
\begin{equation}\label{eq:har}
H_p=\sum_{l=1}^p\frac{1}{p}=\log p + \gamma+\zeta_p
\end{equation}
in $(b)$ is the $p\ts{th}$ harmonic number. The object $\gamma$ in \eqref{eq:har} is the Euler–Mascheroni constant, and $\zeta_p	 = \frac{1}{2p} - \mathcal{O}(\frac{1}{p^4})$ is a monotonically decreasing sequence related to Hurwitz zeta function that satisfies $\zeta_p-\zeta_{p-q} = \frac{1}{2p} -\frac{1}{2(p-q)} + \mathcal{O}(\frac{1}{(p-q)^4})$ for all integers $p > q$ \cite{lang2013algebraic}.
Therefore, using the identity \eqref{eq:har} and the fact that $(1+x)^y\leq e^{xy}$ for any real 
number $y>0$, we obtain
\begin{equation}
\begin{aligned}
\Pr\{\Phi = \emptyset\}&\stackrel{}{\leq} \left((1-\frac{s}{N})e^{\frac{s}{2N(N-s)}}\right)^{|\O\backslash \S|}.
\end{aligned}
\end{equation}
Let $\beta_1 = 1 +(\frac{s}{2N}-\frac{1}{2(N-s)})$. Employing the inequality $\log (1-x)\leq -x-\frac{x^2}{2}$ for $0<x<1$ yields $\Pr\{\Phi = \emptyset\} \leq e^{-\frac{\beta_1 s}{N}{|\O\backslash \S|}}$.
Following a similar argument one can obtain $\Pr\{\Phi = \emptyset\} \leq e^{-\frac{ s}{N}{|\O\backslash \S|}}$. 

Let $\beta = \max\{1,\beta_1\}$. Then
\begin{equation}\label{eq:pbound}
\Pr\{\Phi \neq \emptyset\} \geq 1- e^{-\frac{\beta s}{N}|\O\backslash \S|}\geq \frac{1-\epsilon^\beta}{m}(|\O\backslash \S|)
\end{equation}
from the definition of $s= \frac{N}{m}\log (1\slash \epsilon)$ and the fact that $1- e^{-\frac{\beta s}{N}x} $ is a concave function. According to Lemma 2 in \cite{mirzasoleiman2014lazier},
\begin{equation}
\E[f_{(i+1)_s}(\S)|\S]\geq \frac{\Pr\{\Phi \neq \emptyset\}}{|\O\backslash \S|}\sum_{j\in \O\backslash \S}f_o(\S).
\end{equation}
Hence, 
\begin{equation}\label{eq:lazy}
\E\left[f_{(i+1)_s}(\S)|\S\right]\geq \frac{1-\epsilon^{\beta}}{m}\sum_{j\in \O\backslash \S}f_j(\S).
\end{equation}

On the other hand, employing Lemma \ref{lem:curv} with $\T=\O\cup \S$ and invoking 
monotonicity of $f$ yields
\begin{equation}
\begin{aligned}
\frac{f(\O)-f(\S)}{C(r)}&\leq \frac{f(\O\cup \S)-f(\S)}{C(r)}\leq  \sum_{j\in \O\backslash \S}f_j(\S)\\
& \leq \frac{m}{1-\epsilon^{\beta}}\E\left[f_{(i+1)_s}(\S)|\S\right],
\end{aligned}
\end{equation}
where $|\O\backslash \S|=r$. Let $c=\max\{{\cal C}_{f},1\}$. Applying the law of total expectation and the fact that $C(r)\leq c$	yields
\begin{equation}
\E\left[f(\S\cup \{(i+1)_s\})-f(\S)\right]\geq \frac{1-\epsilon^\beta}{mc}\left(f(\O)-\E\left[f(\S)\right]\right).
\end{equation}
With the established result, the proof simplifies to that of the classical greedy algorithm \cite{nemhauser1978analysis}. Therefore, by using a simple inductive argument,
\begin{equation}
\begin{aligned}
\E[f(\S_{rg})]&\geq \left(1-\left(1-\frac{1-\epsilon^\beta}{mc}\right)^m\right)f(\O)\\
&\stackrel{}{\geq} \left(1-e^{-\frac{1}{c}}-\frac{\epsilon^\beta}{c}\right)f(\O) = \alpha f(\O),
\end{aligned}
\end{equation}
where the last inequality is due to the facts that  $(1+x)^y\leq e^{xy}$ for $y>0$ and $e^{ax}\leq 1+axe^a$ for $0<x<1$. Finally, the stated result follows by using the definition of $f(\S)$.
This completes the proof.
\section{Proof of Theorem \ref{thm:pac}}\label{pf:pac}
Consider the $i\ts{th}$ iteration of Algorithm \ref{alg:greedy}. Let $\S$ denote the current 
sampling set and let $(i+1)_{g}$ and $(i+1)_{rg}$ denote indices of the nodes selected at 
the $(i+1)\ts{st}$ iteration of the greedy sampling algorithm 
\cite{chamon2017greedy,shamaiah2010greedy,shamaiah2012greedy} and Algorithm 
\eqref{alg:greedy}, respectively. Similar to the proof of Theorem \eqref{thm:exp}, we start 
by reducing the proof to that of the classical greedy algorithm. To this end, we employ 
Lemma \ref{lem:curv} with $\T=\O\cup \S$ and use monotonicity of $f$ to obtain
\begin{equation}
\begin{aligned}
f(\O)-f(\S)\leq f(\O\cup \S)-f(\S)\leq  c \sum_{j\in \O\backslash \S}f_j(\S).
\end{aligned}
\end{equation}
Note that given the current sampling set $\S$, from the selection criteria of greedy and 
randomized-greedy algorithms for all $j$ it follows that
\begin{equation}\label{eq:rel1}
\begin{aligned}
f(\O)-f(\S)\leq cmf_{(i+1)_{g}}(\S),
\end{aligned}
\end{equation}
where we used the fact that $|\O\backslash \S|\leq m$. On the other hand, 
\begin{equation}\label{eq:rel2}
\begin{aligned}
f(\S\cup\{(i+1)_{rg}\}) - f(\S)&=f_{(i+1)_{rg}}(\S) \\
&=\eta_{i+1} f_{(i+1)_{g}}(\S).
\end{aligned}
\end{equation} 
Combining \eqref{eq:rel1} and \eqref{eq:rel2} yields
\begin{equation}
f(\S\cup\{(i+1)_{rg}\}) - f(\S) \geq \frac{\eta_{i+1}}{mc}\left(f(\O)-f(\S)\right).
\end{equation}
Using a similar inductive argument as we did in the proof of Theorem \ref{thm:exp} and 
due to the fact that $(1+x)^y\leq e^{xy}$ for $y>0$, it follows that
\begin{equation}\label{eq:pacb1}
\begin{aligned}
f(\S_{rg}) &\geq  \left(1-\left(1- \sum_{i=1}^m\frac{\eta_i}{mc}\right)\right)f(\O)\\
&\stackrel{}{\geq} \left(1- e^{-\sum_{i=1}^m\frac{\eta_i}{mc}}\right)f(\O).
\end{aligned}
\end{equation}
Note that if we assume $\{\eta_i\}$ are independent, the term $\sum_{i=1}^m\eta_i$ is a sum of independent bounded random variables. Since $\{\eta_i\}$ are bounded random variables, by Popoviciu's inequality \cite{hogg1995introduction} for all $i\in[m]$ it holds that $\Var[\eta_i] \leq \frac{1}{4}$. Therefore, using Bernstein's inequality\cite{hogg1995introduction} it holds that for all $0<q<1$
\begin{equation}\label{eq:pber}
\Pr\{\sum_{i=1}^m\eta_i< (1-q)m\mu_{\epsilon}\} \leq e^{-\frac{m(1-q)^2\mu_{\epsilon}^2}{\frac{1-q}{3}\mu_{\epsilon}+\frac{1}{4}}} = e^{-C(\epsilon,q)m}.
\end{equation}
Employing this results in \eqref{eq:pacb1} yields 
\begin{equation}
f(\S_{rg}) \geq\left(1- e^{-\frac{(1-q)\mu_{\epsilon}}{c}}\right)f(\O),
\end{equation}
with probability at least $1-e^{C(\epsilon,q)m}$. Recalling the definition of $f(\S)$ leads to the 
stated bound which in turn completes the proof.
\section{Proof of Theorem \ref{thm:curv}}\label{pf:curv}
To prove the stated result, we begin by exploiting the recursive formulation of the marginal gain 
derived in Proposition \ref{thm:p} to establish a sufficient condition for weak submodularity of 
$f(\S)$. More specifically, from the definition of the maximum element-wise curvature and 
\eqref{eq:mg}, for all $(\S,\T,j)\in \mathcal{X}_l$ we have
\begin{equation}
\begin{aligned}
\frac{f_j(\T)}{f_j(\S)}=\frac{(\u_j^\top\bar{\Sb}_\T^{2}\u_j)(\sigma_j^{2}+\u_j^\top\bar{\Sb}_\S\u_j)}{(\u_j^\top\bar{\Sb}_\S^{2}\u_j)(\sigma_j^{2}+\u_j^\top\bar{\Sb}_\T\u_j)}.
\end{aligned}
\end{equation}	
Next, we employ Courant–Fischer min-max theorem \cite{bellman1997introduction} to obtain
\begin{equation} \label{eq:curv40}
\begin{aligned}
\frac{f_j(\T)}{f_j(\S)}&\leq \frac{\lambda_{max}(\bar{\Sb}_\T^{2})(\sigma_j^{2}+\lambda_{max}(\bar{\Sb}_\S)\|\u_j\|_2^2)}{\lambda_{min}(\bar{\Sb}_\S^{2})(\sigma_j^{2}+\lambda_{min}(\bar{\Sb}_\T)\|\u_j\|_2^2)}\\
&\stackrel{(a)}{\leq}\frac{\lambda_{max}(\bar{\Sb}_\T^{2})(\sigma_j^{2}+\lambda_{max}(\bar{\Sb}_\S))}{\lambda_{min}(\bar{\Sb}_\S^{2})(\sigma_j^{2}+\lambda_{min}(\bar{\Sb}_\T))},
\end{aligned}
\end{equation}
where $(a)$ holds since
\begin{equation}
g(x) = \frac{\sigma_j^{2}+\lambda_{max}(\bar{\Sb}_\S)x}{\sigma_j^{2}+\lambda_{min}(\bar{\Sb}_\T)x}
\end{equation} 
is a monotonically increasing function for $x>0$ and $\|\u\|_2^2 \leq 1$. Given the fact that $\lambda_{max}(\bar{\Sb}_\S)=\lambda_{min}(\bar{\Sb}_\S^{-1})^{-1}$, \eqref{eq:curv40} simplifies to 
\begin{equation} \label{eq:curv41}
\begin{aligned}
\frac{f_j(\T)}{f_j(\S)}{\leq}\frac{\lambda_{min}(\bar{\Sb}_\T^{-1})^{-2}(\sigma_j^{2}+\lambda_{min}(\bar{\Sb}_\S^{-1})^{-1})}{\lambda_{max}(\bar{\Sb}_\S^{-1})^{-2}(\sigma_j^{2}+\lambda_{max}(\bar{\Sb}_\T^{-1})^{-1})}.
\end{aligned}
\end{equation}
By Weyl's inequality \cite{bellman1997introduction}, for all $(\S,\T,j)\in \mathcal{X}_l$ it holds that  $\lambda_{min}(\bar{\Sb}_{\mathcal{N}}^{-1})\geq \lambda_{min}(\bar{\Sb}_\T^{-1}) \geq \lambda_{min}(\bar{\Sb}_\S^{-1}) \geq \lambda_{min}(\P^{-1})$ and $\lambda_{max}(\bar{\Sb}_{\mathcal{N}}^{-1})\geq \lambda_{max}(\bar{\Sb}_\T^{-1}) \geq \lambda_{max}(\bar{\Sb}_\S^{-1}) \geq \lambda_{max}(\P^{-1})$. Hence, by definition of maximum element-wise curvature  we have
\begin{equation}
\begin{aligned}
\mathcal{C}_{\max}&\leq\max_{j \in \mathcal{N}} \frac{\lambda_{max}(\P)^{2}(\sigma_j^{2}+\lambda_{max}(\P))}{\lambda_{max}(\bar{\Sb}_{\mathcal{N}}^{-1})^{-2}(\sigma_j^{2}+\lambda_{max}(\bar{\Sb}_{\mathcal{N}}^{-1})^{-1})}\\
&\stackrel{(a)}{\leq}\max_{j \in \mathcal{N}}\frac{(\sigma_j^{2}+\lambda_{max}(\P))(\lambda_{min}(\P)^{-1}+\sigma_j^{-2})^{2}}{\lambda_{max}(\P)^{-2}(\sigma_j^{2}+(\lambda_{min}(\P)^{-1}+\sigma_j^{-2})^{-1})},
\end{aligned}
\end{equation}
where $(a)$ follows since $\lambda_{max}(\bar{\Sb}_{\mathcal{N}}^{-1})\leq \lambda_{max}(\P^{-1}+\sigma_j^{-2}\I_N)$ and because the maximum eigenvalue of a positive definite matrix satisfies the triangle inequality. Note that the denominator of the last inequality is always strictly larger than 
$\sigma_j^2$, and that $\lambda_{max}(\P)\geq \lambda_{min}(\P)$. Following some straight-forward algebra, we obtain
\begin{equation}
\mathcal{C}_{\max} \leq\max_{j \in \mathcal{N}} \frac{\lambda_{\max}^2(\P)}{\lambda_{\min}^2(\P)}\left(1+\frac{\lambda_{\max}(\P)}{\sigma_j^2}\right)^3
\end{equation}
which is the stated result. This completes the proof.
\end{appendices}
\bibliographystyle{ieeetr}
\bibliography{refs}
\end{document}